\documentclass[letter,11pt]{article}
\usepackage[margin=1in]{geometry}
\usepackage[utf8]{inputenc}

\usepackage{times}
\usepackage{soul}
\usepackage{url}
\usepackage[hidelinks]{hyperref}
\usepackage[utf8]{inputenc}
\usepackage[small]{caption}
\usepackage{graphicx}
\usepackage{amsmath}
\usepackage{amsthm}
\usepackage{booktabs}
\usepackage{csquotes}
\usepackage[switch]{lineno}

\usepackage{tikz}
\usepackage{caption,subcaption}
\usepackage{thmtools,thm-restate} 
\usepackage{cleveref}
\usepackage{todonotes}
\usepackage{amssymb}

\usepackage{xspace}
\usepackage[vlined,linesnumbered,ruled]{algorithm2e}
 	
\numberwithin{equation}{section}

\usepackage{graphics,graphicx,xcolor,color,url}    
\usetikzlibrary{arrows,shapes,snakes,automata,backgrounds,petri}
\tikzset{
	cross/.style={cross out, draw=black, minimum size=2*(#1-\pgflinewidth), inner sep=0pt, outer sep=0pt},
	cross/.default={3pt},
	vertex/.style={circle, draw, fill=black, inner sep=0pt, minimum width=4pt},
}

\newcommand{\drawpred}[2]{
	\draw (#1, #2) node[cross,red] {};
}
\newcommand{\drawreal}[2]{
	\draw[black!40!green] (#1, #2) circle (3pt);
}
\newcommand{\interval}[4]{
	\draw (#2, #4) node[anchor=east]{#1} -- (#3, #4);
	\draw (#2, #4-0.1) -- (#2, #4+0.1);
	\draw (#3, #4-0.1) -- (#3, #4+0.1);
}
\newcommand{\intervalp}[5]{
	\interval{#1}{#2}{#3}{#4}
	\drawpred{#5}{#4}
}
\newcommand{\intervalr}[5]{
	\interval{#1}{#2}{#3}{#4}
	\drawreal{#5}{#4}
}
\newcommand{\intervalpr}[6]{
	\interval{#1}{#2}{#3}{#4}
	\drawpred{#5}{#4}
	\drawreal{#6}{#4}
}

\SetKw{KwLet}{let}
\SetKw{KwAnd}{and}

\newtheorem{theorem}{Theorem}[section]
\newtheorem{lemma}[theorem]{Lemma}
\newtheorem{coro}[theorem]{Corollary}

\newtheorem{claim}[theorem]{Claim}
\newtheorem{definition}[theorem]{Definition}

\usepackage{thm-restate}

\newcommand{\lw}{0.5mm}

\author{
	Thomas Erlebach\thanks{Department of Computer Science, Durham University, United Kingdom, \texttt{thomas.erlebach@durham.ac.uk}. Supported by EPSRC grants EP/S033483/2 and EP/T01461X/1} \and
	Murilo Santos de Lima\thanks{Munich, Germany, \texttt{mslima@ic.unicamp.br}. Funded by EPSRC grant EP/S033483/1.}
	\and
	Nicole Megow\thanks{Faculty of Mathematics and Computer Science, University of Bremen, Germany, \texttt{\{nmegow,jschloet\}@uni-bremen.de}. Supported by the German Science Foundation (DFG) under contract ME~3825/1.} \and Jens Schl{\"o}ter\footnotemark[3] 
}

\newcommand{\jo}{\ensuremath{k^+}\xspace} 
\newcommand{\oj}{\ensuremath{k^-}\xspace} 

\newcommand{\ALG}{\mathrm{ALG}}
\newcommand{\OPT}{\mathrm{OPT}}
\newcommand{\opt}{|\OPT|}

\newcommand{\w}{\ensuremath{\overline{w}}} 
\newcommand{\pred}[1]{\overline{#1}} 

\newcommand{\ZZ}{\mathbb{Z}}

\newcommand{\RR}{\mathbb{R}}

\newcommand{\eps}{\ensuremath{\varepsilon}\xspace} 
\newcommand{\sym}{\ensuremath{\Delta}\xspace} 
\DeclareMathOperator{\EX}{\mathbb{E}}
\newcommand{\ud}{\ensuremath{D}}
\newcommand{\hs}{\ensuremath{\mathcal{H}}}

\newcommand{\nnew}[1]{#1}

\newcommand{\jnew}[1]{#1}

\graphicspath{ {figs/} }

\title{Sorting and Hypergraph Orientation under Uncertainty with  Predictions}

\begin{document}

\maketitle

\begin{abstract}
	Learning-augmented algorithms have been attracting increasing interest,	but have only recently been considered in the setting of explorable uncertainty where precise values of uncertain input elements can be obtained by a query and the goal is to minimize the number of queries needed to solve a problem. We study learning-augmented algorithms for sorting and hypergraph orientation under uncertainty, assuming access to untrusted predictions for the uncertain values. Our algorithms provide improved performance guarantees for accurate predictions while maintaining worst-case guarantees that are best possible without predictions. For hypergraph orientation, for any $\gamma \geq 2$, we give an algorithm that achieves a competitive ratio of $1+1/\gamma$ for correct predictions and $\gamma$ for arbitrarily wrong predictions. For sorting, we achieve an optimal solution for accurate predictions while still being $2$-competitive for arbitrarily wrong predictions. These tradeoffs are the best possible. We also consider different error metrics and show that the performance of our algorithms degrades smoothly with the prediction error in all the cases
	where this is possible.
\end{abstract} 

\section{Introduction}

The emerging research area of learning-augmented algorithm design has been attracting increasing
attention in recent years. For {\em online} algorithms, it was initiated in~\cite{lykouris2018competitive} for caching and has fostered an overwhelming number of results for numerous problems, including 
online graph problems~\cite{KumarPSSV19,lindermayrMS22,EberleLMNS22,AzarPT22}
and scheduling problems~\cite{purohit2018improving,angelopoulos2019online,Mitzenmacher20,LattanziLMV20,AzarLT21,AzarLT22,LindermayrM22,BampisDKLP22,LiX21}. In the research area of
\emph{explorable uncertainty}
\cite{kahan91queries,
erlebach15querysurvey
}, however,
learning-augmented algorithms have only 
recently been studied for the first time,
namely for the minimum spanning tree problem (MST) \cite{EdLMS22}. 
This area considers problems with uncertainty in the input data assuming that a \emph{query} to an input element can be used to obtain the exact value of that element. Queries are costly, and hence the goal is to make as few queries as possible
until sufficient information has been obtained to solve the given
problem.
Learning-augmented algorithms in this setting are given (untrusted) predictions
of the uncertain input values. In this paper, we present and analyze learning-augmented
algorithms for further fundamental problems 
with explorable uncertainty,
namely sorting and hypergraph orientation. 


In the \emph{hypergraph orientation problem under uncertainty}~\cite{BampisDEdLMS21}, we are given a hypergraph $H=(V,E)$. 
Each vertex $v \in V$ is associated with an \emph{uncertainty interval} $I_v=(L_v,U_v)$ and an, initially unknown, \emph{precise weight} $w_v \in I_v$.
We call $L_v$ and $U_v$ the \emph{lower and upper limit} of~$v$.
A query of $v$ reveals its precise
{weight} $w_v$ and reduces its uncertainty interval to $I_v=[w_v]$.
Our task is to orient each hyperedge $S \in E$ towards the vertex of minimum precise weight in $S$.
An adaptive algorithm
can sequentially make queries to
vertices
to learn their weights
until it has enough information to identify the minimum-weight vertex of each
hyperedge.
A set $Q \subseteq V$ is called \emph{feasible} if querying $Q$ reveals sufficient information to find the orientation.
As queries are costly, the goal is to (adaptively) find a feasible query set of minimum size.
An important special case is  {when the input graph is a simple graph that 
is exactly the interval graph induced by the uncertainty intervals~$\mathcal{I}$}. 
This special case corresponds to the problem of sorting  {a set of unknown values represented by}
uncertainty intervals and, therefore,  
we refer to it as \emph{sorting 
under uncertainty}.

Since there exist input instances that are impossible to solve without querying all vertices, we evaluate our algorithms in an {\em instance-dependent} manner: For each input, we compare the number of queries made by an algorithm with the 
  {minimum} possible number of queries \emph{for that input}, using \emph{competitive analysis}. 
For a given problem instance, let $\OPT$ denote an optimal query set. An algorithm is $\rho$-{\em competitive} if it executes, for any problem instance, at most $\rho \cdot {\opt}$ queries.
As the query results are uncertain and, to a large extent, are the deciding factor whether querying certain vertices is a good strategy or not, the problem has a clear online flavor. 
In particular, the uncertainty prevents $1$-competitive algorithms, even without any running time  {restrictions}.

Variants of hypergraph orientation 
have been  {widely} studied since the model of explorable uncertainty has been proposed~\cite{kahan91queries}. 
Sorting and hypergraph orientation
are well known to admit efficient polynomial-time
algorithms if precise input data is given, and they
are well understood in the setting of explorable uncertainty: The best
known deterministic algorithms are $2$-competitive, and no deterministic algorithm
can be better \cite{kahan91queries,halldorsson19sortingqueries,BampisDEdLMS21}.
For sorting, the competitive ratio can be improved to~$1.5$ using randomization~\cite{halldorsson19sortingqueries}.
In the stochastic setting, where the precise weights of the vertices are drawn according to known distributions over the intervals, there exists a $1.618$-competitive algorithm for hypergraph orientation and a $4/3$-competitive algorithm for special cases~\cite{BampisDEdLMS21}.
For the stochastic sorting problem, the algorithm with the optimal expected query cost is known, but its competitive ratio remains unknown~\cite{ChaplickHLT21}.
Besides hypergraph orientation and other selection-type problems, there has been work on
combinatorial optimization problems with explorable uncertainty, such as shortest path~\cite{feder07pathsqueires},  knapsack~\cite{goerigk15knapsackqueries}, scheduling problems~\cite{DurrEMM20,arantes18schedulingqueries,albersE2020,AlbersE21,GongGLM22}, the  {minimum spanning tree} (MST) problem and matroids~\cite{erlebach08steiner_uncertainty,erlebach14mstverification,megow17mst,FockeMM20,MerinoS19}. 

In this paper, we consider a third model (to the adversarial and stochastic setting) and assume that the algorithm has, for each vertex $v$, access to a prediction $\pred{w}_v \in I_v$ for the unknown weight $w_v$. 
These predicted weights could for example be computed using machine learning~(ML) methods or other statistical tools on past data.
Given the emerging success of ML methods, it seems reasonable to expect predictions of high accuracy. 
However, there are no theoretical guarantees and the predictions might be arbitrarily wrong, which raises the question whether an ML algorithm performs sufficiently well in all circumstances.
In the context of hypergraph orientation and explorable uncertainty, we answer this question affirmatively by designing learning-augmented algorithms that perform very well if the predictions are accurate and still match the 
  {adversarial} lower bounds even if the predictions are arbitrarily wrong.
To formalize these properties, we use the notions of $\alpha$-consistency
and $\beta$-robustness~\cite{lykouris2018competitive,purohit2018improving}: 
An algorithm is $\alpha$-consistent if it
is $\alpha$-competitive when the predictions are correct, and it is
$\beta$-robust if it is $\beta$-competitive no matter how wrong the
predictions are. 
Additionally, we aim at guaranteeing a smooth transition between consistency and robustness by giving performance guarantees that gracefully degrade with the \emph{amount} of prediction error.
This raises interesting questions regarding appropriate ways of
measuring prediction errors, and we explore several such measures.
Analyzing algorithms in terms of error-dependent consistency and robustness allows us to still give worst-case guarantees (in contrast to the stochastic setting) that are more fine-grained than guarantees in the pure adversarial setting. 

  {Research on}
explorable uncertainty  {with} untrusted predictions is in its infancy.
In the only previous work that we are aware of, Erlebach et al.~\cite{EdLMS22} studied the MST  {problem} with uncertainty and untrusted predictions. They  {propose an error measure $k_h$ called {\em hop-distance} (definition follows later)} and obtain  {a}
$1.5$-consistent and $2$-robust algorithm as well as a parameterized consistency-robustness tradeoff
with smooth degradation, 
$\min\{1+ \frac{1}{\gamma} + \frac{5 \cdot k_h}{\opt}, \gamma+1 \}$, for any integral $\gamma \geq 2$ and error $k_h$.  
  {It remained open whether or how other problems with explorable uncertainty can leverage untrusted predictions, and what other prediction models or error metrics might be useful in explorable uncertainty settings.} 
%
%

\paragraph*{Main results} 


We show how to  
utilize  {possibly erroneous} predictions for hypergraph orientation and sorting with explorable uncertainty.
For sorting, we present an algorithm
that is $1$-consistent and $2$-robust, which is  {in a remarkable way} best possible:
the algorithm identifies an optimal query set if the predictions are accurate,
while maintaining the best possible worst-case ratio of~$2$  {for arbitrary predictions}.
For hypergraph orientation, we give a $1.5$-consistent and $2$-robust algorithm and show that this consistency
is best possible when aiming for optimal robustness.

Our major focus lies on a more fine-grained performance analysis with guarantees that improve with the prediction accuracy. 
{A key ingredient in this line of research is the choice of error measure quantifying this (in)accuracy.} 
We 
  {propose and discuss} three different  {error} measures $k_{\#}, k_h$ and $k_M$: 
The number of inaccurate predictions $k_{\#}$ is natural and allows  {a} smooth degradation result for
sorting, but in general  {it} is too crude to allow  
improvements upon lower bounds of~$2$ for the setting without predictions.
We therefore also consider the \emph{hop distance}~$k_h$, as proposed in~\cite{EdLMS22}.  {Further, we introduce} a new error measure $k_M$ called \emph{mandatory query distance} which is tailored to problems with explorable uncertainty. 
It is defined in a more problem-specific way, and 
we show it to be more restrictive {in the sense that}~$k_M\leq k_h$.

%


%
For the sorting problem, we obtain an algorithm with competitive ratio $\min\{1+k/\opt, 2\}$, where $k$ can be any of the three error measures considered, which is best possible.
For the hypergraph orientation problem, we provide an algorithm with competitive ratio $\min\{(1+\frac{1}{\gamma-1})(1+ \frac{k_M}{\opt}), \gamma \}$, for any integral $\gamma\geq 2$. This is best possible for $k_M=0$ and large $k_M$. With respect to the hop distance, we achieve the stronger bound $\min\{(1+\frac{1}{\gamma})(1+ \frac{k_h}{\opt}), \gamma\}$, for any integral $\gamma \geq 2$, which is also best possible for $k_h = 0$ and large $k_h$.
While the  {consistency and robustness trade-off} of the $k_M$-dependent algorithm is weaker, we remark that the corresponding algorithm requires less predicted information than the $k_h$-dependent algorithm  {and that the error dependency can be stronger as $k_M$ can be significantly smaller than $k_h$.}
  {Further, we note that the parameter $\gamma$ is in both results restricted to integral values since it determines sizes of query sets, \nnew{but} 
  a generalization to reals $\gamma\geq 2$ is possible at a small loss in the guarantee. For $k_h$ this follows from arguments given in~\cite{EdLMS22}; for $k_M$ this can be shown similarly.}

  {While} our algorithm for sorting has polynomial running time, the algorithms for the hypergraph orientation problem may involve solving an NP-hard vertex cover problem. We justify this  {increased} complexity by showing that even the offline version of the problem (determining the optimal query set if the precise weights 
  are known) is NP-hard.
  
  Finally, we discuss the learnability of the predictions and prove PAC-learnability.
\section{Preliminaries and error measures} 
\label{sec:overview}

An algorithm  {for the hypergraph orientation problem with uncertainty} that 
  {assumes} that the predicted weights are accurate 
can exploit the characterization of optimal solutions given in~\cite{BampisDEdLMS21} to compute a query set that is optimal under this assumption.
Querying this set leads to $1$-consistency but 
  {may perform arbitrarily bad} in case of incorrect predictions (as shown down below).
On the other hand, known $2$-competitive algorithms for the adversarial problems without predictions~\cite{kahan91queries,halldorsson19sortingqueries} are not better than $2$-consistent, and the algorithms for the stochastic setting~\cite{BampisDEdLMS21} do not guarantee any robustness at all.
The known lower bounds of $2$ rule out any robustness factor less than $2$ for our model. They build on the simple example with a single edge $\{u,v\}$ and intersecting intervals $I_v,I_u$. No matter which vertex a deterministic algorithm queries first, say $v$, the realized weight could be $w_v\in I_v\cap I_u$, which requires a second query. If the adversary chooses $w_u\notin I_v\cap I_u$, querying just $u$ would have been sufficient to identify the vertex of minimum weight. 

The following bound on the best~achievable~tradeoff between consistency and robustness translates from the lower bound in~\cite{EdLMS22} for MST under uncertainty with predictions. 
  {Later in this paper}, we provide algorithms with matching performance guarantees.
Note that this lower bound does not hold for the sorting problem.

\begin{restatable}{theorem}{ThmLBTradeoffWithoutError}
\label{theo_minimum_combined_lb}
  Let 
  $\beta \geq 2$ be a fixed integer.
  For 
  hypergraph orientation 
  under uncertainty, there is no deterministic $\beta$-robust algorithm that is $\alpha$-consistent for $\alpha < 1 + \frac{1}{\beta}$. And vice versa, no deterministic $\alpha$-consistent algorithm, with $\alpha>1$, is $\beta$-robust for $\beta < \max\{\frac{1}{\alpha-1},2\}$.
  The result holds even for orienting a single hyperedge or a simple (non-hyper) graph.
\end{restatable}

The following proof is an adjusted variant of the proof in~\cite{EdLMS22} of the analogous lower bound for MST under explorable uncertainty with predictions.

\begin{proof}
	Assume, for the sake of contradiction, that there is a deterministic $\beta$-robust algorithm that is $\alpha$-consistent with $\alpha=1 + \frac{1}{\beta}-\eps$, for some $\eps>0$.
	Consider an instance with vertices $\{0,1,\ldots,\beta\}$, a single hyperedge that contains all $\beta+1$ vertices, and intervals and predicted weights as in Figure~\ref{fig_combined_lb_min_single_set}.
	The algorithm must query the vertices $\{1, \ldots, \beta\}$ first as otherwise it would query $\beta+1$ vertices in case all predictions are correct, while there is an optimal query set of size $\beta$.  
	Suppose w.l.o.g.\ that the algorithm queries the vertices $\{1, \ldots, \beta\}$ in order of increasing indices. Consider the adversarial choice $w_i = \pred{w}_i$, for $i = 1, \ldots, \beta - 1$, and then $w_{\beta} \in I_0$ and $w_0 \notin I_1 \cup \ldots \cup I_{\beta}$.
	This forces the algorithm to query also~$I_0$, while an optimal solution only queries~$I_0$.
	Thus any such algorithm has robustness at least~$\beta+1$, a contradiction.
	
	The second part of the theorem directly follows from the first part and the known general lower bound of $2$ on the competitive ratio~\cite{erlebach08steiner_uncertainty,kahan91queries}. Assume there is an $\alpha$-consistent deterministic algorithm with $\alpha=1 + \frac{1}{\beta'}$ for some integer $\beta'\in [1,\infty)$. Consider the instance above with $\beta=\beta'-1$. Then the algorithm has to query the vertices $\{1, \ldots, \beta\}$ first to ensure $\alpha$-consistency as otherwise it would have a competitive ratio of $\frac{\beta+1}{\beta}>1+\frac{1}{\beta'}=\alpha$ in case that all predictions are correct. By the argumentation above, the robustness factor of the algorithm is at least $\beta+1=\beta'=\frac{1}{\alpha-1}$.
	
	To prove the result for the (non-hyper) graph orientation problem, the only 
	difference is that we use edges $\{0,i\}$ for $1\le i\le \beta$.
\end{proof}

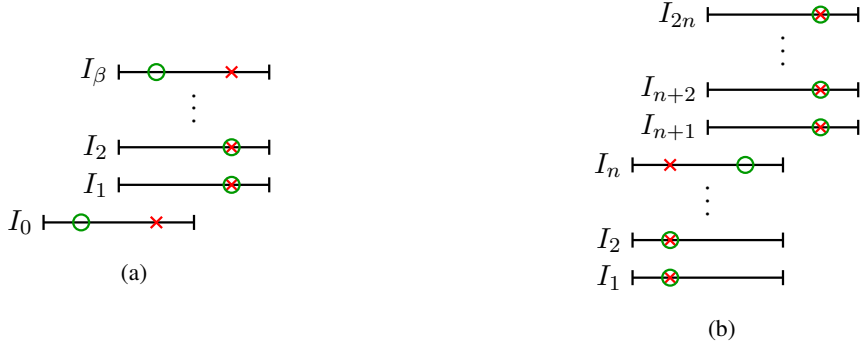
\begin{figure}[t]
	\centering
	\captionsetup[subfigure]{justification=centering}
	\begin{subfigure}{0.45\textwidth}
		\centering
		\begin{tikzpicture}[line width = 0.3mm, scale=1]
			\intervalpr{$I_1$}{1}{3}{0}{2.5}{2.5}
			\intervalpr{$I_2$}{1}{3}{0.5}{2.5}{2.5}
			\intervalpr{$I_\beta$}{1}{3}{1.5}{2.5}{1.5}
			\path (2, 0.5) -- (2, 1.6) node[font=\normalsize, midway, sloped]{$\dots$};
			
			\intervalpr{$I_0$}{0}{2}{-0.5}{1.5}{0.5}
		\end{tikzpicture}
		\subcaption{}\label{fig_combined_lb_min_single_set}
	\end{subfigure}\quad
	\begin{subfigure}{0.45\textwidth}
		\centering
		\begin{tikzpicture}[line width = 0.3mm, scale=1]
			\intervalpr{$I_1$}{0}{2}{0}{0.5}{0.5}
			\intervalpr{$I_2$}{0}{2}{0.5}{0.5}{0.5}
			\intervalpr{$I_n$}{0}{2}{1.5}{0.5}{1.5}
			\path (1, 0.5) -- (1, 1.6) node[font=\normalsize, midway, sloped]{$\dots$};
			
			\intervalpr{$I_{n+1}$}{1}{3}{2}{2.5}{2.5}
			\intervalpr{$I_{n+2}$}{1}{3}{2.5}{2.5}{2.5}
			\intervalpr{$I_{2n}$}{1}{3}{3.5}{2.5}{2.5}
			\path (2, 2.5) -- (2, 3.6) node[font=\normalsize, midway, sloped]{$\dots$};
		\end{tikzpicture}
		\subcaption{}\label{fig_lb_wrong_predictions}
	\end{subfigure}\quad
	\caption{Instances for lower bounds.
		Red crosses indicate predicted weights, and green circles show correct weights.
		(\subref{fig_combined_lb_min_single_set})~Lower bound on robustness-consistency tradeoff.
		(\subref{fig_lb_wrong_predictions}) Lower bound based on the number of inaccurate predictions.
		\label{fig_lowerbounds}
	}
\end{figure}


\subsection{Preliminaries}

The crucial structure and unifying concept in hypergraph orientation and sorting under explorable uncertainty without predictions are {\em witness sets}~\cite{bruce05uncertainty}. Witness sets are the key to any comparison with an optimal solution. A ``classical'' witness set is a set of vertices for which we can guarantee that any feasible solution must query at least {\em one} of these vertices.
In the classical setting without access to predictions, sorting and hypergraph orientation admit $2$-competitive online algorithms that rely essentially on identifying and querying disjoint witness sets of size two.
We refer to witness sets of size two also as {\em witness pairs}.
The following lemma characterizes witness pairs for our problems.
We call a vertex $v$ {\em leftmost} in a hyperedge~$S$ if it is a vertex with minimum lower limit $L_v$ in~$S$.

\begin{lemma}[Kahan, 1991]
	\label{lem:witness}
	Given (hyper)graph $H=(V,E)$.
	Consider some $S \in E$.
	A set $\{v, u\} \subseteq S$ with $I_v \cap I_u \neq \emptyset$, and $v$ or~$u$ leftmost in~$S$, is a witness set.
\end{lemma}

In terms of learning-augmented algorithms, completely relying on querying witness pairs ensures $2$-robustness, but it does not lead to any improvements in terms of consistency.
In order to obtain an improved consistency, we need stronger local guarantees.
To that end, we call a vertex \emph{mandatory}, if it is part of every feasible query set.
Identifying mandatory vertices based on the interval structure alone is not always possible, otherwise there would be a $1$-competitive algorithm. 
Therefore, we want to identify vertices that are mandatory under the assumption that the predictions are correct.
We call such vertices \emph{prediction mandatory}.
The following lemma gives a characterization of (prediction) mandatory vertices.
Recall that the uncertainty interval $I_v$ of a vertex $v$ is either of form $I_v = (L_v, U_v)$ or $I_v= [w_v]$. In the former case, we call $I_v$ \emph{non-trivial} and otherwise we call $I_v$ \emph{trivial}.

\begin{restatable}{lem}{LemmaMandatoryMin}
	\label{lema_mandatory_min}
	A vertex $v \in V$ with a non-trivial uncertainty interval is mandatory if and only if there is a hyperedge $S \in E$ with $v \in S$ such that either $(i)$ $v$ is a minimum-weight vertex of $S$ and $w_u \in I_v$ for some $u \in S \setminus \{v\}$, or $(ii)$ $v$ is not a minimum-weight vertex of $S$ and $w_u \in I_v$ for the minimum-weight vertex $u$ of $S$.
\end{restatable}

\begin{proof}
	If $v$ is a minimum-weight vertex of hyperedge~$S$ and \nnew{$I_v$} contains~$w_u$ of another vertex $u\in S\setminus \{v\}$, then~$S$ cannot be oriented even if we query all vertices in $S \setminus \{v\}$ as we cannot prove $w_v \le w_u$ without querying $v$. 
	If~$v$ is not a minimum-weight vertex of a hyperedge $S$ with $v\in S$ and $I_v$ contains the minimum weight~$w^*$ of $S$, then~$S$ cannot be solved even if we query all vertices in $S \setminus \{v\}$, as we cannot prove that $w^* \leq w_v$ without querying $v$.
	
	If~$v$ is a minimum-weight vertex of hyperedge $S$, but $w_u \notin I_v$ for every $u \in S \setminus \{v\}$, then $S \setminus \{v\}$ is a feasible solution for orienting~$S$.
	If~$v$ is not a minimum-weight vertex of hyperedge~$S$ and $I_v$ does not contain the minimum weight of~$S$, then again $S \setminus \{v\}$ is a feasible solution for~$S$. If every hyperedge $S$ that contains $v$ falls into one of these two cases, then querying all vertices except $v$ is a feasible query set for the whole instance.
\end{proof}

By using the lemma with the predicted weights instead of the precise weights, we can identify prediction mandatory vertices.
Furthermore, the lemma does not only enable us to identify mandatory vertices given full knowledge of the precise weights, but also implies criteria to identify \emph{known mandatory} vertices, i.e., vertices that are known to be mandatory given only the hypergraph, the intervals, and precise weights revealed by previous queries.
Every algorithm can query such vertices without worsening its competitive ratio as formalized by the following corollary. We call a hyperedge \emph{not yet solved}, if we do not know the orientation of the hyperedge yet. In particular, for such hyperedges, no leftmost vertex can have a trivial uncertainty interval.

\begin{restatable}{coro}{CorMinLeftMandatory} \label{cor_min_left_mandatory}
	If the interval $I_v$ of a leftmost vertex~$v$ in a not yet solved hyperedge~$S$ contains the precise weight of another vertex in~$S$, then~$v$ is mandatory.
	In particular, if $v$ is leftmost in $S$ and $I_u \subseteq I_v$ for some $u \in S \setminus \{v\}$, then $v$ is mandatory.
\end{restatable}

We use Lemma~\ref{lema_mandatory_min} to define an offline algorithm, i.e., we assume full access to the precise weights but still want to compute a feasible query set, that follows a two-stage structure:
First, we iteratively query all mandatory vertices computed using Lemma~\ref{lema_mandatory_min}.
After that, each not yet oriented hyperedge~$S$ has the following configuration: The leftmost vertex~$v$ has a precise weight outside $I_u$ for all $u \in S \setminus \{v\}$, and each other vertex in~$S$ has precise weight outside~$I_v$.
Thus we can either query~$v$ or all other vertices $u \in S\setminus \{v\}$ with $I_u \cap I_v \not= \emptyset$ to determine the orientation.
The optimum solution is to query a minimum vertex cover in the following auxiliary graph as introduced in~\cite{BampisDEdLMS21}:

\begin{definition}
	\label{def:vertex_cover_instance}
	Given a hypergraph $H = (V,E)$, the \emph{vertex cover instance} of $H$ is the graph $\bar{G}=(V,\bar{E})$ with
	$\{v,u\} \in \bar{E}$ if and only if there is a not yet solved hyperedge $F \in E$ such that $v,u \in F$, $v$ is leftmost in $F$ and $I_v \cap I_u \neq \emptyset$.
	For the sorting problem, it holds that $\bar{G}=G$.
\end{definition}

The Lemmas~\ref{lem:witness} and~\ref{lema_mandatory_min} directly imply that the offline algorithm is optimal.
The algorithm may require exponential time, but this is not surprising as we also show that the offline version of the hypergraph orientation problem is NP-hard (cf.~\Cref{app:nphard}).


A key idea of our algorithms with access to predictions is to emulate the offline algorithm using the predicted information. Since blindly following the offline algorithm might lead to a competitive ratio of $n$ for faulty predictions, we have to augment the algorithm with additional, carefully selected queries. 

The next lemma formulates a useful property of vertex cover instances without known mandatory vertices. 

\begin{restatable}{lemma}{lemVCProperty}
	\label{lem:VC-property}
	Given a hypergraph $H = (V,E)$ without known mandatory vertices (by~\Cref{cor_min_left_mandatory}), let $Q$ be an arbitrary vertex cover of $\bar{G}$.
	After querying $Q$, for each hyperedge $F \in E$, we either know the orientation of $F$ or can determine it by exhaustively querying according to~\Cref{cor_min_left_mandatory}.
\end{restatable} 

\begin{proof}
	Consider an arbitrary hyperedge $S$ and let $v$ be leftmost in $S$.
	As the instance does not contain vertices that are known mandatory by~\Cref{cor_min_left_mandatory}, the leftmost vertex is unique.
	If the interval of $v$ is not intersected by any $I_u$ with $u \in S\setminus \{v\}$ before querying $Q$, then we clearly already know the orientation. Thus, assume otherwise.
	By assumption, we have $I_u \not\subseteq I_v$ and $I_v \not\subseteq I_u$ for all $u \in S \setminus \{v\}$ (before querying $Q$).
	Thus, a vertex cover of the vertex cover instance $\bar{G}$ contains either (i) $v$, (ii) $S \setminus \{v\}$, or (iii) $S$. If the vertex cover $Q$ contains the complete hyperedge $S$, then, after querying it, we clearly know the orientation of~$S$. 
	
	It remains to argue about (i) and (ii). 
	Consider case (i) and let $u$ be leftmost in $S \setminus \{v\}$. If querying $v$ reveals $w_v \not\in I_u$, then $v$ has minimum weight in $S$ and we know the orientation of $S$.
	If querying $v$ reveals \jnew{$w_v \in I_u$}, then the query reduces $I_v$ to $[w_v]$ and $u$ becomes mandatory by~\Cref{cor_min_left_mandatory} as $I_v = [w_v] \jnew{\subseteq I_u}$. 
	After querying $u$ (to exhaustively apply~\Cref{cor_min_left_mandatory}), we can repeat the argument with the vertex leftmost in $S \setminus \{u,v\}$.
	
	Consider case (ii). 
	If querying $S \setminus \{v\}$ reveals that there exists no vertex $u \in S\setminus \{v\}$ with $w_u \in I_v$, then $v$ must be of minimum weight in $S$.
	Otherwise, the uncertainty interval of some $u \in S \setminus \{v\}$ was reduced to $[w_v] \subseteq I_v$ and we can apply~\Cref{cor_min_left_mandatory}. 
	After that, all elements of $S$ are queried and we clearly know the orientation.
\end{proof}

\subsection{Accuracy of predictions}

{While 
consistency and robustness only consider the extremes in terms of prediction quality, we aim for a more fine-grained analysis that relies on error metrics to measure the quality of the predictions.
As was discussed in~\cite{EdLMS22}, for the MST problem, 
simple error measures like the number of inaccurate predictions $k_\#=|\{v \in V\,|\, w_v \not= \w_v\}|$ or an $\ell_1$ error metric such as $\sum_{e\in E}|w_e - \w_e|$ are not meaningful;  this is also true for the hypergraph orientation problem. 
In particular, we show that even for $k_{\#} = 1$  the competitive ratio cannot be better than the known bound of~$2$ for general hypergraph orientation. 

\begin{theorem}\label{theo_lb_wrong_predictions}
	If $k_{\#} \geq 1$, then any deterministic algorithm for the hypergraph orientation problem under uncertainty with predictions has competitive ratio $\rho\geq 2$. The result holds even for (non-hyper) graphs.
\end{theorem}

The following proof is an adjusted variant of the proof in~\cite{EdLMS22} of the analogous lower bound for MST under explorable uncertainty with predictions.

\begin{proof}
	Consider a hypergraph with vertices $\{1,\ldots,2n\}$, hyperedges 
	$S_i=\{i, n+1, n+2, \ldots, 2n\}$, for $i = 1, \ldots, n$, and intervals and predicted weights as depicted in Figure~\ref{fig_lb_wrong_predictions}.
	Assume w.l.o.g.\ that the algorithm queries the vertices $\{1,\ldots,n\}$ in the order $1, 2, \ldots, n$ and the vertices $\{n+1,\ldots,2n\}$ in the order $n+1, n+2, \ldots, 2n$.
	Before the algorithm queries vertex~$n$ or~$2n$, the adversary sets all predictions as correct, so the algorithm will eventually query~$n$ or~$2n$.
	If the algorithm queries~$n$ before~$2n$, then the adversary chooses a weight $w_n$ for $n$ that forces a query to all vertices $n+1, \ldots, 2n$, and the predicted weights for those vertices as correct, so the optimum solution only queries $n+1, \ldots, 2n$. This situation is illustrated in Figure~\ref{fig_lb_wrong_predictions}.
	
	A symmetric argument holds if the algorithm queries vertex~$2n$ before vertex~$n$.
	
	For (non-hyper) graphs, we use edges
	$\{i,j\}$ for all $1\le i\le n, n+1\le j\le 2n$ instead.
\end{proof}

For the sorting problem, however, we show that $k_{\#}$-dependent guarantees are indeed possible.
In general we need more refined measures that take the interleaving  structure of intervals into account.} 

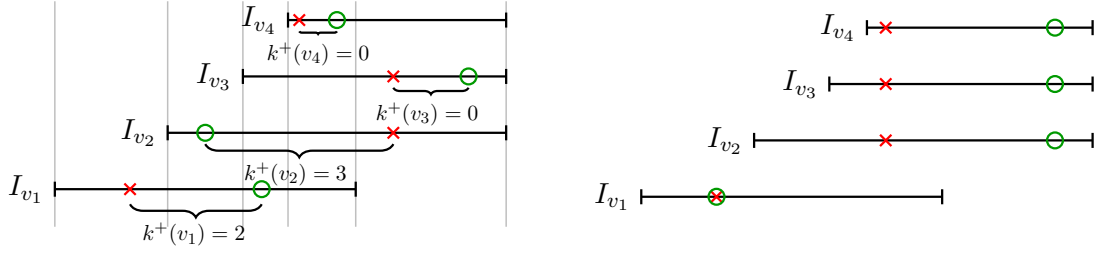
\begin{figure*}[t]
	\centering
	\begin{subfigure}[l]{0.4\textwidth}
		\begin{tikzpicture}[line width = 0.3mm, scale = 1, transform shape]	
		\intervalpr{$I_{v_1}$}{0}{4}{1.5}{1}{2.75}	
		\intervalpr{$I_{v_2}$}{1.5}{6}{2.25}{4.5}{2}	
		\intervalpr{$I_{v_3}$}{2.5}{6}{3}{4.5}{5.5}							
		\intervalpr{$I_{v_4}$}{3.1}{6}{3.75}{3.25}{3.75}		
		
		\begin{scope}[on background layer]
		\draw[line width = 0.2mm,,lightgray] (0,1) -- (0,4);
		\draw[line width = 0.2mm,,lightgray] (1.5,1) -- (1.5,4);
		\draw[line width = 0.2mm,,lightgray] (2.5,1) -- (2.5,4);
		\draw[line width = 0.2mm,,lightgray] (3.1,1) -- (3.1,4);
		\draw[line width = 0.2mm,,lightgray] (6,1) -- (6,4);
		\draw[line width = 0.2mm,,lightgray] (4,1) -- (4,4);
		\end{scope}
		
		\draw[decoration= {brace, amplitude = 5 pt, aspect = 0.5}, decorate] (2.75,1.3) -- (1,1.3);
		\node[] (l1) at (1.85,0.9) {\scalebox{0.75}{$\jo(v_1)=2$}};
		
		\draw[decoration= {brace, amplitude = 5 pt, aspect = 0.5}, decorate] (4.5,2.1) -- (2,2.1);
		\node[] (l1) at (3.2,1.7) {\scalebox{0.75}{$\jo(v_2)=3$}};
		
		\draw[decoration= {brace, amplitude = 2.5 pt, aspect = 0.5}, decorate] (5.5,2.85) -- (4.5,2.85);
		\node[] (l1) at (4.95,2.5) {\scalebox{0.75}{$\jo(v_3)=0$}};
		
		\draw[decoration= {brace, amplitude = 1 pt, aspect = 0.5}, decorate] (3.75,3.6) -- (3.25,3.6);
		\node[] (l1) at (3.5,3.3) {\scalebox{0.75}{$\jo(v_4)=0$}};

		\end{tikzpicture}
	\end{subfigure}
\hspace*{1cm}
\begin{subfigure}[r]{0.4\textwidth}
	\begin{tikzpicture}[line width = 0.3mm, scale = 1, transform shape]
	\intervalpr{$I_{v_1}$}{0}{4}{1.5}{1}{1}		
	\intervalpr{$I_{v_2}$}{1.5}{6}{2.25}{3.25}{5.5}		
	\intervalpr{$I_{v_3}$}{2.5}{6}{3}{3.25}{5.5}								
	\intervalpr{$I_{v_4}$}{3}{6}{3.75}{3.25}{5.5}				
	
	\node[] (l1) at (1.85,0.95) {};	
	\end{tikzpicture}
\end{subfigure}
	\caption{Example from~[Erlebach \emph{et al.}, 2022] interpreted for the hypergraph orientation problem with a hyperedge $S=\{v_1,v_2,v_3,v_4\}$. Circles illustrate precise weights and crosses illustrate the predicted weights. Predictions and precise weights with a total hop distance of $k_h = 5$ and mandatory query distance of $k_M = 1$ (left) and $k_h=3$ and $k_M = 1$ (right).}
	\label{fig_error_ex} 
\end{figure*}

As a first refined measure, we consider the {\em hop distance} as proposed in~\cite{EdLMS22}.
For a vertex $v$ and any vertex $u \in V \setminus \{v\}$, we define the function $k_{u}(v)$ that indicates whether the relation of $w_v$ to interval $I_u$ changes compared to the relation of $\w_v$ and $I_u$. 
To be more precise, $k_u(v) = 1$ if $\w_v \le L_u < w_v$, $w_v \le L_u < \w_v$, $w_v < U_u \le \w_v$ or $\w_v < U_u \le w_v$, and $k_u(v)= 0$ otherwise.
Based on this function, we define the hop distance of a single vertex as $\jo(v) = \sum_{u \in V\setminus \{v\}} k_{u}(v)$. 
Intuitively $\jo(v)$ for a single $v \in V$ counts the number of relations between $w_v$ and intervals $I_{u}$ with $u \in V\setminus \{v\}$ that are not accurately predicted.
For a set of vertices $V' \subseteq V$, we define $\jo(V') = \sum_{v \in V'} \jo(v)$.
Finally, we define the hop distance by $k_h=\jo(V)$.
For an example see Figure~\ref{fig_error_ex}.

Note that $k_{\#} = 0$ implies $k_h = 0$, so Theorem~\ref{theo_minimum_combined_lb} implies that no algorithm can simultaneously have competitive ratio better than $1 + \frac{1}{\beta}$ if $k_h = 0$ and $\beta$ for arbitrary~$k_h$.

While the hop distance takes the interval structure into account, it does not distinguish whether a ``hop'' affects a feasible solution. We introduce a third and strongest error measure based on the sets of (prediction) mandatory vertices.

Let $\mathcal{I}_P$ be the set of prediction mandatory vertices, 
and let $\mathcal{I}_R$ be the set of really mandatory vertices. 
The 
\emph{mandatory query distance} is the size of the symmetric difference of $\mathcal{I}_P$ and $\mathcal{I}_R$, i.e., $k_M = |\mathcal{I}_P \sym \mathcal{I}_R| = |(\mathcal{I}_P \cup \mathcal{I}_R) \setminus (\mathcal{I}_P \cap \mathcal{I}_R)| = |(\mathcal{I}_P \setminus \mathcal{I}_R) \cup (\mathcal{I}_R \setminus \mathcal{I}_P)|$.
Figure~\ref{fig_error_ex}  {(right)} shows an example with $k_M=1$.
Considering the precise weights in the example, both $\{v_1\}$ and $\{v_2,v_3,v_4\}$ are feasible solutions. Thus, no element is part of every feasible solution and $\mathcal{I}_R = \emptyset$. 
Assuming correct predicted weights,  {we have that} $v_1$ has to be queried even if all other vertices have already been queried and, therefore, $\mathcal{I}_P = \{v_1\}$.
It follows $k_M = |\mathcal{I}_P \sym \mathcal{I}_R| = 1$. 

Obviously, $k_M$ is a problem-specific error measure as, in a given set of uncertainty intervals, different intervals may be mandatory for different problems.
We can relate $k_M$ to $k_h$.
 
\begin{restatable}{theorem}{HopDistanceMandatoryDistance}
	\label{Theo_hop_distance_mandatory_distance}
	For any instance of hypergraph orientation  {under uncertainty with predictions}, the hop distance is at least as large as the mandatory query distance, i.e., $k_M \leq k_h$.
\end{restatable}

\begin{proof}
	Consider an instance with uncertainty intervals~$\mathcal{I}$, 
	precise weights $w$ and predicted weights $\pred{w}$. 
	Recall that $\mathcal{I}_P$ and
	$\mathcal{I}_R$ are the sets of prediction mandatory vertices and mandatory vertices, respectively.
	Observe that $k_M$ counts the vertices that
	are in $\mathcal{I}_P\setminus \mathcal{I}_R$ and those that are in $\mathcal{I}_R\setminus \mathcal{I}_P$.
	We will show the claim that, for
	every vertex $v$ in those sets, there is a vertex $u$ such that
	the weight of $u$ passes over $L_v$ or $U_v$ (or both) when going
	from $\w_v$ to $w_v$.
	This means
	that each vertex $v \in \mathcal{I}_P \sym \mathcal{I}_R$ is mapped to a unique
	pair $(u,v)$ such that the weight of $u$ passes over at least one endpoint of~$I_v$, and
	hence each such pair contributes at least one to~$k_h$. This implies $k_M \le k_h$.
	
	It remains to prove the claim. Consider a $v\in \mathcal{I}_P\setminus \mathcal{I}_R$. (The argumentation
	for intervals in $\mathcal{I}_R\setminus \mathcal{I}_P$ is symmetric, with the roles of $w$ and $\w$ exchanged.)
	As $v$ is not in $\mathcal{I}_R$, replacing all intervals for vertices
	in $\mathcal{I}\setminus\{v\}$ by their precise weights yields an instance
	that is solved.
	This means that in every hyperedge $S\in E$ that contains $v$, one of the following
	cases holds:
	\begin{itemize}
		\item[(a)] $v$ is known not to be the minimum of $S$ {w.r.t.\ precise weights $w$}. It follows that there
		is a vertex $u$ in $S$ with $w_u\le L_v$.
		\item[(b)] $v$ is known to be the minimum of $S$ {w.r.t.\ precise weights $w$}. It follows that all
		vertices $u\in S\setminus\{v\}$ satisfy $w_u\ge U_v$.
	\end{itemize}
	As $v$ is in $\mathcal{I}_P$, replacing all intervals
	of vertices in $V\setminus\{v\}$ by their predicted weights yields an instance
	that is not solved.
	This means that there exists at least one hyperedge $S'\in E$ that contains $v$ and satisfies
	the following:
	\begin{itemize}
		\item[(c)] All vertices $u$ in $S'\setminus\{v\}$ satisfy $\w_u>L_v$, and there
		is at least one such $u$ with $L_v<\w_u<U_v$.
	\end{itemize}
	
	If $S'$ falls into case (a) above, then by (a) there is a vertex
	$u$ in $S'\setminus \{v\}$ with $w_u\le L_v$, and by (c) we have $\w_u>L_v$. This
	means that the weight of $u$ passes over~$L_v$.
	If $S'$ falls into case (b) above, then by (c) there exists an
	vertex $u$ in $S'\setminus \{v\}$ with $\w_u < U_v$, and by (b) we have $w_u\ge U_v$.
	Thus, the weight of $u$ passes over~$v$. This establishes the claim,
	and hence we have shown that $k_M\le k_h$ for the hypergraph orientation problem. 
\end{proof}  

%

\section{Hypergraph orientation}
\label{sec:optimal-tradeoff}

We consider the general hypergraph orientation problem, and give error-sensitive algorithms w.r.t.~the measures $k_h$ and $k_M$.
For~$k_h$, we show that we can use an adjusted variant of the algorithm given in~\cite{EdLMS22} for the MST problem.

To achieve error-sensitive guarantees w.r.t~$k_M$, we show that the $k_h$-dependent guarantee does not directly translate to $k_M$. Instead, we give an even simpler algorithm that emulates the offline algorithm and augments it with additional queries, and show that this algorithm achieves the optimal trade-off w.r.t~$k_M$. We remark that this algorithm only requires access to the set of prediction mandatory vertices, which is a weaker type of prediction than access to the predicted weights $\w$.

\subsection{Error-sensitive algorithm w.r.t.\ hop distance}

We start by designing an algorithm that achieves the optimal consistency and robustness tradeoff (cf.~\Cref{theo_minimum_combined_lb}) with a linear error dependency on the hop distance $k_h$. The following theorem summarizes our main result with respect to error measure $k_h$. We achieve the theorem by adjusting the algorithm and analysis given in~\cite{EdLMS22} for the MST problem.  In contrast to the result of~\cite{EdLMS22}, we achieve the optimal consistency-robustness tradeoff by essentially exploiting~\Cref{lem:VC-property}. The corresponding analysis requires additional work specific to the hypergraph orientation problem.

\begin{restatable}{theorem}{thmminimumhop}
	\label{thm_minimum:hop}
	There is an algorithm for hypergraph orientation under explorable uncertainty with predictions that, given $\gamma \in \mathbb{N}_{\geq 2}$, achieves a competitive ratio of $\min\{ (1 + \frac{1}{\gamma})(1 + k_h/\opt), \gamma\}$.
	If $\gamma = 2$, then the competitive ratio is $\min\{1.5 + k_h/\opt, 2\}$.
\end{restatable}

\subsubsection*{Overview of the algorithmic ideas} 

To match the lower bound on the consistency and robustness tradeoff, we, for a given $\gamma \in \mathbb{N}_{\ge 2}$, have to guarantee $(1+\frac{1}{\gamma})$-consistency and $\gamma$-robustness.
As we observed before, following the offline algorithm based on the predicted information can lead to an arbitrarily bad robustness while using just the witness set algorithm will not improve upon $2$-consistency.
The idea of our algorithm (and the one given in~\cite{EdLMS22}) is to emulate the offline algorithm using the predicted information and to combine it with the witness set algorithm.

To illustrate this idea, assume $\gamma = 2$, so we have to guarantee $1.5$-consistency and $2$-robustness.
To combine both algorithms, we consider \emph{strengthened witness sets}~\cite{EdLMS22}, which are sets $W \subseteq V$ of size three such that every feasible solution queries at least \emph{two} members of $W$. If we repeatedly query strengthened witness sets, then we achieve a competitive ratio of at most $1.5$, matching our target consistency. Clearly, strengthened witness sets cannot always be identified based on the given graph and intervals alone. If we always could identify a strengthened witness set when the instance is not solved yet, then we would have a $1.5$-competitive algorithm by repeatedly querying such sets, contradicting the lower bound of $2$ on the adversarial competitive ratio. Therefore, we have to identify strengthened witness sets based on the predicted information, i.e., we identify sets $W$ of cardinality three such that each feasible solution contains at least two members of the set \emph{if} the predictions of the elements in $W$ are correct. Since we can only identify strengthened witness sets based on the predicted information, we cannot afford to just query the complete set, as we might lose the guarantee on the set if the predictions are faulty, which could violate our target $2$-robustness. To that end, we query such sets in a carefully selected order that allows us to detect errors that might cause us to violate the $2$-robustness after at most two queries within the set. We select the first two queries within the set in such a way that they form a witness set. So even if there is an error within these queries, we can discard the third query and the two queries we already executed will never violate $2$-robustness as they form a witness set. Furthermore, we will show that we can charge all executed queries that violate the consistency bound to a distinct error.

Our algorithm does this repeatedly until we cannot identify strengthened witness sets anymore, not even by using the predictions. 
After that, the instance has a certain structure that allows us to solve it with an adjusted second phase of the offline algorithm while achieving the optimal consistency and robustness tradeoff of~\Cref{theo_minimum_combined_lb} with linear error dependency on $k_h$.

If $\gamma > 2$, then the first phase of our algorithm repeatedly identifies a strengthened witness set \emph{and} $\gamma-2$ prediction mandatory vertices. It then queries the $\gamma-2$ prediction mandatory vertices and afterwards proceeds to query the strengthened witness sets as described above. We show that this adjustment allows us to achieve the optimal tradeoff with linear dependency on $k_h$ for every integral $\gamma \ge 2$.

\subsubsection*{Precise algorithm and formal analysis}

As mentioned above, our algorithm is based on identifying strengthened witness sets using the predicted weights, i.e., sets $W \subseteq V$ with $|W|= 3$ such that every feasible solution contains at least two elements of $W$ if the predictions are accurate. We want to be able to query $W$ in such a way that (i) the first two queries in $W$ are a witness set and (ii) after the first two queries in $W$ we either have detected a prediction error or can guarantee that each feasible solution indeed contains two elements of $W$ (no matter if the predicted weights of vertices outside $W$ are correct or not). 

To achieve this, we identify prediction mandatory vertices that are not only mandatory if all predicted weights are correct but become mandatory if a \emph{single} predicted weight is correct. To that end, we use the following definition.

\begin{definition}
	We say that a predicted weight~$\pred{w}_u$ {\em enforces} another vertex~$v$ if $u$ and $v$ have non-trivial uncertainty intervals, $\pred{w}_u \in I_v$, and $u,v \in S$, where~$S$ is a hyperedge such that either~$v$ is leftmost in~$S$, or~$u$ is leftmost in~$S$ and~$v$ is leftmost in $S \setminus \{u\}$.
\end{definition}

If the predicted weight $\w_u$ of a vertex $u$ enforces another vertex $v$ and the predicted weight of $u$ is accurate, then after querying $u$ we know for sure that $v$ is indeed mandatory. The following lemma formulates this property.

\begin{lemma} \label{lemma_min_kh_enforce_mandatory}
	If $\pred{w}_u$ enforces~$v$, then $\{v, u\}$ is a witness set.
	Also, if $w_u \in I_v$, then $v$ is mandatory.
\end{lemma}

\begin{proof}
	Since $\pred{w}_u$ enforces~$v$, there must be a hyperedge $S$ with $v,u \in S$ such that $\pred{w}_u \in I_v$ and either $v$ is leftmost in $S$ or $u$ is leftmost in $S$ and $v$ is leftmost in $S \setminus \{u\}$.
	The first claim follows from Lemma~\ref{lem:witness} as one of $u,v$ is leftmost in $S$ and $I_v \cap I_u \not= \emptyset$.

	If~$v$ has minimum precise weight in~$S$ or is leftmost in~$S$, then the second claim follows from Lemma~\ref{lema_mandatory_min} and Corollary~\ref{cor_min_left_mandatory}.
	Otherwise, the fact that $w_u \in I_v$ and that~$v$ is leftmost in $S \setminus \{u\}$ implies that~$I_v$ contains the minimum precise weight, so the claim follows from Lemma~\ref{lema_mandatory_min}.
\end{proof}

We can now define our Algorithm~\ref{ALG_min_beta}. Within the definition of the algorithm, the \emph{current instance} always refers to the problem instance obtained after executing all previous queries.
We say that a vertex is prediction mandatory for the current instance, if it is mandatory if the predicted weight of all \emph{not yet queried} vertices are correct.

The algorithm ensures that the instance never contains known mandatory vertices according to~\Cref{cor_min_left_mandatory} by exhaustively querying those vertices (cf.~Lines~\ref{line_min_beta_kh_mandatory_first},~\ref{line_min_beta_kh_mandatory_loop} and~\ref{line_min_beta_kh_mandatory_second}).
The Lines~\ref{lin_min_beta_cond_trio} to~\ref{line_min_beta_query_first} identify and query strengthened witness sets. This is done by identifying a witness set $\{u,w\}$ such that the predicted weight $\w_u$ enforces another vertex $v$. Only if $w_u \in I_v$, the algorithm also queries $v$. If that is the case, then~\Cref{lemma_min_kh_enforce_mandatory} implies that $v$ is mandatory and, therefore, every feasible solution must query at least two members of $\{u,v,w\}$. If such a triple $\{u,v,w\}$ of vertices does not exist but there still is a vertex $u$ such that $\w_u$ enforces a vertex $v$, then the algorithm just queries $v$ in Line~\ref{line_min_beta_query_first}. We will prove that this never violates the target consistency or robustness.
If $\gamma > 2$, then the algorithm precedes this step by querying (up-to) $\gamma - 2$ predictions mandatory vertices in Lines~\ref{line_min_beta_kh_calc_predict_first} to~\ref{line_min_beta_kh_calc_predict_second}. The algorithm does this repeatedly as long as possible and afterwards queries a minimum vertex cover of the current vertex cover instance in Line~\ref{p2-line2-kh}. Exploiting~\Cref{lem:VC-property}, the remaining instance can then be solved by exhaustively querying vertices that are mandatory due to~\Cref{cor_min_left_mandatory} (cf.~Line~\ref{p2-line3-kh}).

\begin{algorithm}[tb]
	\KwIn{Hypergraph $H=(V,E)$, intervals $I_v$ and predictions $\w_v \in I_v$ for all $v \in V$}
	\label{line_min_beta_start}
	\lWhile{there is a known mandatory vertex~$v$ by~\Cref{cor_min_left_mandatory}}{query $v$\label{line_min_beta_kh_mandatory_first}}
	\Repeat{the current instance has no prediction mandatory vertices \label{line_min_beta_kh_cond_loop}}{
		$Q \leftarrow \emptyset$\; 
		$P \leftarrow$ set of prediction mandatory vertices for the current instance\; \label{line_min_beta_kh_calc_predict_first}
		\While{$P \neq \emptyset$ \KwAnd $|Q| < \gamma - 2$ \label{line_while}}{
			pick and query some $u \in P$; \quad $Q \leftarrow Q \cup \{u\}$\; \label{line_min_beta_kh_predict}
			
			\lWhile{there is a known mandatory vertex~$v$ by~\Cref{cor_min_left_mandatory}}{query $v$\label{line_min_beta_kh_mandatory_loop}}
			$P \leftarrow$ set of prediction mandatory vertices for  {the current} instance\; \label{line_min_beta_kh_calc_predict_second}
		}
		\uIf{$\exists \mbox{ distinct } u,v,w$ s.t.~$\pred{w}_u$ enforces $v$ and $\{u,w\}$ is a witness set for the current instance by \Cref{lem:witness} \label{lin_min_beta_cond_trio}} {
			query $u, w$\; \label{line_min_beta_wit_trio}
			\lIf{$w_u \in I_v$}{query $v$ \label{line_min_beta_mand_trio}}
		} \lElseIf{$\exists v,u$ such that $\pred{w}_u$ enforces $v$ \label{line_min_beta_pair}} {
			query $v$ \label{line_min_beta_query_first}
		}
		\lWhile{there is a known mandatory vertex~$v$ by~\Cref{cor_min_left_mandatory}}{query $v$\label{line_min_beta_kh_mandatory_second}}
	}
	Compute and query a minimum vertex cover~$Q'$ for the current vertex cover instance\label{p2-line2-kh}\; 
	\lWhile{there is a known mandatory vertex~$v$ by~\Cref{cor_min_left_mandatory}}{query $v$\label{p2-line3-kh}}
	\caption{Learning-augmented algorithm for the hypergraph orientation problem under explorable uncertainty with respect to the hop distance $k_h$}
	\label{ALG_min_beta}
\end{algorithm}

We proceed by proving three more important properties of the algorithm that will help us to prove~\Cref{thm_minimum:hop}. 

The~\Cref{lemma_min_kh_enforce_mandatory} implies that if $\pred{w}_u$ enforces~$v$ and the predicted value of $u$ is correct, then $v$ is mandatory. This directly implies that $v$ is prediction mandatory for the current instance: If the predicted weight of all not yet queried vertices are correct, then the predicted weight of $u$ is correct and $v$ is mandatory. This leads to the following corollary.

\begin{coro}
	\label{lemma_min_beta_kh_enforce_predmand}
	Consider a point of execution of the algorithm in which a predicted weight~$\pred{w}_u$ enforces another vertex~$v$.
	It holds that~$v$ is prediction mandatory for the current instance.
\end{coro}

Next, we show that there is at most one iteration of the repeat-loop that executes less than $\gamma-2$ queries in Line~\ref{line_min_beta_kh_predict} or no queries in Lines~\ref{line_min_beta_wit_trio}--\ref{line_min_beta_query_first}.

\begin{lemma}
	\label{lemma_min_beta_kh_unique_loop}
	Every iteration of the repeat-loop in Algorithm~\ref{ALG_min_beta} (cf.~Lines~$1$--\ref{line_min_beta_kh_cond_loop}) apart from the final one executes $\gamma-2$ queries in Line~\ref{line_min_beta_kh_predict} and at least one query in Lines~\ref{line_min_beta_wit_trio}--\ref{line_min_beta_query_first}. Furthermore, if an iteration executes a query in Lines~\ref{line_min_beta_wit_trio}--\ref{line_min_beta_query_first}, then it executes $\gamma-2$ queries in Line~\ref{line_min_beta_kh_predict}.
\end{lemma}

\begin{proof}
	Consider an iteration of the repeat-loop that executes less than $\gamma-2$ queries in Line~\ref{line_min_beta_kh_predict}. 
	Then, the while-loop from Line~\ref{line_while} to Line~\ref{line_min_beta_kh_calc_predict_second} terminates because the current instance has no prediction mandatory vertices.
	This means that the algorithm also does not execute any queries in Lines~\ref{line_min_beta_wit_trio}--\ref{line_min_beta_query_first} as there cannot be a predicted weight $\w_u$ that enforces another vertex $v$ because $v$ would be prediction mandatory for the current instance by~\Cref{lemma_min_beta_kh_enforce_predmand}.
	Since the algorithm does not execute queries in Lines~\ref{line_min_beta_wit_trio}--\ref{line_min_beta_query_first}, the current instance still does not contain prediction mandatory vertices at the end of the current iteration of the repeat-loop, which implies that the loop terminates.
	
	To conclude the proof, consider an iteration of the repeat-loop that executes $\gamma-2$ queries in Line~\ref{line_min_beta_kh_predict} but no queries in Lines~\ref{line_min_beta_wit_trio}--\ref{line_min_beta_query_first}.
	No queries in Lines~\ref{line_min_beta_wit_trio}--\ref{line_min_beta_query_first} imply that there is no $\w_u$ that enforces a vertex $v$.
	Since the current instance contains no vertices that are mandatory by~\Cref{cor_min_left_mandatory}, this means that for each hyperedge $S$ we have that (i) $\w_u \not\in I_v$ for the unique leftmost vertex $v$ in $S$ and all not yet queried vertices $u \in S\setminus \{v\}$ and (ii) $\w_v \not\in I_u$ for the unique leftmost vertex $v$ in $S$ and all not yet queried vertices $u \in S\setminus \{v\}$. If the predictions of the not yet queried vertices are correct, then we can find the orientation of each hyperedge $S$ by either querying the unique leftmost vertex $v$ in $S$ or all not yet queried vertices in $S\setminus\{u\}$. This implies that no vertex is prediction mandatory for the current instance and, therefore, the loop terminates.
\end{proof}

We prove the following lemma that helps us to prove that queries in Line~\ref{line_min_beta_pair} will never violate our target consistency or robustness.

\begin{lemma}
	\label{lemma_min_no_witness_implies_solved}
	Let $u, v$ be a pair that satisfies the condition in Line~\ref{line_min_beta_pair} of Algorithm~\ref{ALG_min_beta} leading to a query of~$v$.
	After querying $v$, vertex $u$ will either become mandatory by~\Cref{cor_min_left_mandatory} and be queried in the next execution of Line~\ref{line_min_beta_kh_mandatory_second} or for each hyperedge $S$ containing $u$ we either know the orientation $S$ or know that $u$ cannot be of minimum weight in $S$.
\end{lemma}

\begin{proof}
	Consider the instance before~$v$ is queried.
	Due to the failed test in Line~\ref{lin_min_beta_cond_trio}, for every hyperedge~$S$ containing~$v$, the following facts hold:
	\begin{enumerate}
		\item If~$u$ is leftmost in~$S$, then the orientation of $S$ is already known, or $v \in S$ and $v$ is the only vertex in~$S \setminus \{u\}$ with an interval that intersects~$I_u$.
		\item If~$u$ is not leftmost in~$S$ but intersects the interval of the leftmost vertex $v'$ in~$S$, then $v = v'$.
		\item If~$u$ is not leftmost in~$S$ and $I_u$ does not intersect the interval of the leftmost vertex in~$S$, then~$u$ is certainly not of minimum weight in~$S$.
	\end{enumerate}
	If condition~(1) holds and the orientation of~$S$ is not known then, after querying~$v$, either $w_v \notin I_u$ and the orientation of~$S$ is determined, or $w_v \in I_u$ and $u$ becomes mandatory by~\Cref{cor_min_left_mandatory}.
	
	If condition~(2) holds, then either $w_v \not\in I_u$ and 
	$u$ is certainly not of minimum weight in $S$, or $w_v \in I_u$ and $u$ becomes mandatory due to~Corollary~\ref{cor_min_left_mandatory}.
	The result follows trivially if condition~(3) holds.
\end{proof}

Intuitively, the lemma means that if vertex $u$ does not become mandatory by~\Cref{cor_min_left_mandatory} after querying $v$, then the algorithm will \emph{never} even consider vertex $u$ anymore as all hyperedges containing $u$ are either resolved or have an orientation that is completely independent of vertex $u$. If that is the case, then the algorithm queries \emph{exactly} one vertex of the witness set $\{u,v\}$. This can never lead to a violation of the target consistency and robustness as even the optimal solution has to query at least one member of $\{u,v\}$. 
If on the other hand $u$ becomes mandatory, then either the predicted weight $\w_u$ is correct and $v$ is also mandatory by~\Cref{lemma_min_kh_enforce_mandatory} or the predicted weight of $\w_u$ is wrong. In the former case even $\OPT$ queries $u$ and $v$, so queries to those vertices certainly do not lead to a violation of the target consistency. In the latter case, $\{u,v\}$ is still a witness set and we will show that we can charge one of the queries against a prediction error caused by vertex $u$, so queries to $\{u,v\}$ do not violate robustness or error-dependency.

Using these insights, we are finally ready to prove~\Cref{thm_minimum:hop}.

\begin{proof}[Proof of~\Cref{thm_minimum:hop}]
	Before we prove the performance bounds, we remark that the algorithm clearly solves the given instance by definition of the final two lines of the algorithm and~\Cref{lem:VC-property}.
	Next, we separately show that the algorithm executes at most $\gamma \cdot \opt$ queries and at most $(1 + \frac{1}{\gamma})(1 + \frac{k_h}{\opt}) \cdot \opt$ queries.
	
	\paragraph{Proof of $|\ALG| \le \gamma \cdot \opt$ (robustness).} We start by proving the robustness bound.
	Vertices queried in Line~\ref{line_min_beta_mand_trio} are mandatory due to~\Cref{lemma_min_kh_enforce_mandatory} and, thus, in any feasible solution. Clearly, querying those vertices will never worsen the competitive ratio of $\ALG$.
	To analyze all further queries executed by $\ALG$, fix an optimum solution $\OPT$.
	
	Consider an iteration of the repeat-loop in which some query is performed in Lines~\ref{line_min_beta_wit_trio}--\ref{line_min_beta_pair}.
	Let~$P'$ be the set of vertices queried in Lines~\ref{line_min_beta_kh_predict},~\ref{line_min_beta_wit_trio} and~\ref{line_min_beta_query_first}.
	If the iteration queries a vertex $v$ in Line~\ref{line_min_beta_query_first} that is enforced by the predicted weight $\w_u$ of a vertex $u$, then we include $u$ in $P'$ independent of whether the algorithm queries $u$ at some point or not.
	Note that, by~\Cref{lemma_min_no_witness_implies_solved}, such a vertex $u$ is considered in exactly one iteration as it either is queried directly afterwards in Line~\ref{line_min_beta_kh_mandatory_second} or will never be considered again be the algorithm (as we argued above). 
	Using this and the fact that we never query vertices multiple times, we can conclude that the sets $P'$ of different iterations are pairwise disjoint.
	We continue by showing that all such sets $P'$ are also witness sets of size at most $\gamma$ and, thus, querying them never violates the $\gamma$-robustness.
	
	By~\Cref{lemma_min_kh_enforce_mandatory}, $P'$ contains exactly $\gamma-2$ vertices queried in Line~\ref{line_min_beta_kh_predict}.
	Furthermore, $P'$ contains either two vertices $u$ and $w$ queried in Line~\ref{line_min_beta_wit_trio} or two vertices $u$ and $v$ as considered in Line~\ref{line_min_beta_query_first}. Either way, we have $|P'| = \gamma$.
	
	Next, we argue that $P'$ is a witness set.
	If Line~\ref{line_min_beta_wit_trio} is executed, then note that $\{u, w\} \subseteq P'$ is a witness set.
	If a query is performed in Line~\ref{line_min_beta_query_first}, then note that $\{v, u\}\subseteq P'$ is a witness set.
	In both cases, $P'$ contains a witness set and, therefore, is a witness set itself.
	We conclude that $P'$ is a witness set of size $\gamma$ and, thus, querying $P'$ never worsens the competitive ratio below~$\gamma$.
	
	Let $V'$ be the set of unqueried vertices in Line~\ref{line_min_beta_kh_calc_predict_first} during the iteration of the repeat-loop consisting of Lines~1--\ref{line_min_beta_kh_cond_loop} in which no query is performed in Lines~\ref{line_min_beta_wit_trio}--\ref{line_min_beta_pair}. 
	Recall that Lemma~\ref{lemma_min_beta_kh_unique_loop} states that there is at most one such iteration and it has to be the last iteration of the loop.
	If no such iteration exists, then let $V'$ denote the set of unqueried vertices before Line~\ref{p2-line2-kh}.
	
	If the orientation is not yet known at this point, then the instance is not yet solved and we have $|\OPT \cap V'| \geq 1$.
	Furthermore, $|Q| \leq \gamma - 2$ holds for the set of queries executed in the iteration of the repeat-loop in which no query is performed in Lines~\ref{line_min_beta_wit_trio}--\ref{line_min_beta_pair}.
	This implies $|Q| \leq (\gamma - 2) \cdot |\OPT \cap V'|$. 
	
	Let $Q'$ denote the set of all vertices queried in Lines~\ref{p2-line2-kh}~and~\ref{p2-line3-kh}. Since the queries of Line~\ref{p2-line2-kh} are a minimum vertex cover for the current instance, they are a lower bound on $|(\OPT \cap V')\setminus Q|$ by~\Cref{lem:witness}. Additionally, all queries of Line~\ref{p2-line3-kh} are mandatory and thus their number is at most $\opt$. This implies that $|Q'| \le 2 \cdot |(\OPT \cap V') \setminus Q|$.
	Combining the bounds for $|Q|$ and $|Q'|$, we get $|Q| + |Q'| \leq \gamma \cdot |\OPT \cap V'|$.
	
	All remaining vertices queried by $\ALG$ have been queried in  Lines~\ref{line_min_beta_kh_mandatory_first},~\ref{line_min_beta_kh_mandatory_second} and~\ref{line_min_beta_kh_mandatory_loop}. 
	Thus, they are mandatory, part of any feasible solution, and never violate the $\gamma$-robustness. This concludes the proof of the robustness bound.
	
	\paragraph{Proof of $|\ALG| \le (1 + \frac{1}{\gamma})(1 + \frac{k_h}{\opt}) \cdot \opt$ (consistency and error-dependency).} 
	We continue by proving consistency and linear error-dependency.
	Fix an optimum solution $\OPT$.
	Let $\oj(u)$ be the number of vertices $v$ such that $v, u \in S$ for some hyperedge $S$, and the value of $v$ passes over an endpoint of $u$, i.e., $w_v \le L_u < \w_v$, $\w_v \le L_u < w_v$, $\w_v < U_u \le w_v$ or $w_v < U_u \le \w_v$.
	From the arguments in the proof of Theorem~\ref{Theo_hop_distance_mandatory_distance}, it can be seen that, for each vertex~$u$ that is prediction mandatory at some point during the execution of the algorithm (not necessarily for the initially given instance) and is not in $\OPT$, we have that $\oj(u) \geq 1$.
	Similar, if $u$ at some point during the execution is not prediction mandatory for the current instance but turns out to be mandatory for the precise weights, then $\oj(u) \geq 1$.
	For a subset $U \subseteq V$, let $\oj(U) = \sum_{u \in U} \oj(u)$.
	Note that $k_h = \oj(V)$ holds by reordering summations.
	
	In the following, we will show for various disjoint subsets $S \subseteq V$ that $|S \cap \ALG| \le (1+\frac{1}{\gamma})\cdot (|\OPT \cap S| + \oj(S))$.
	The union of the subsets $S$ will contain $\ALG$, so it is clear that the bound of $(1+\frac{1}{\gamma}) \cdot (1 + \frac{k_h}{\opt})$ on the competitive ratio of the algorithm follows.
	
	Vertices queried in Line~\ref{line_min_beta_kh_mandatory_loop} are in any feasible solution, so the set~$P_0$ of these \jnew{vertices} satisfies $|P_0| \leq |\OPT \cap P_0|$.
	
	If there is an execution of the loop consisting of Lines~1--\ref{line_min_beta_kh_cond_loop} that does not perform queries in Lines~\ref{line_min_beta_wit_trio}--\ref{line_min_beta_pair}, then let~$P_1$ be the set of vertices queried in Line~\ref{line_min_beta_kh_predict}.
	Every vertex~$u \in P_1$ is prediction mandatory, so if $u \notin \OPT$ then $\oj(u) \geq 1$.
	Thus we have that $|P_1| \leq |P_1 \cap \OPT| + \oj(P_1)$.
	
	Let $V'$ be the set of unqueried vertices before the execution of Line~\ref{p2-line2-kh}, and let $Q'$ denote the vertex cover of Line~\ref{p2-line2-kh}.
	Since $Q'$ is a minimum vertex cover of the vertex cover instance, we have that $|Q'| \leq |\OPT \cap V'|$.
	Let~$M$ be the set of vertices queried in Line~\ref{p2-line3-kh}.
	Each vertex $v \in M$ is known mandatory because it contains the precise weight~$w_u$ of a vertex $u \in Q$.
	But since $v$ was not prediction mandatory before querying $Q'$, we have $\pred{w}_u \notin I_v$ and, therefore, $\oj(v) \geq 1$.
	This implies $|V' \cap \ALG| = |Q' \cup M| \leq |V' \cap \OPT| + \oj(M) \leq |V' \cap \OPT| + \oj(V')$.
	
	Finally, consider an execution of the repeat-loop in which some query is performed in Lines~\ref{line_min_beta_wit_trio}--\ref{line_min_beta_pair}.
	Let~$Q$ be the set of vertices queried in Line~\ref{line_min_beta_kh_predict}, and let~$W$ be the set of vertices queried in Lines~\ref{line_min_beta_wit_trio}--\ref{line_min_beta_query_first}.
	If a query is performed in Line~\ref{line_min_beta_query_first} and $u$ is queried in Line~\ref{line_min_beta_kh_mandatory_second} directly afterwards, then include~$u$ in~$W$ as well.
	Note that $|Q| = \gamma - 2$ holds by~\Cref{lemma_min_beta_kh_unique_loop}. If~$\pred{w}_u$ enforces~$v$ in Line~\ref{lin_min_beta_cond_trio} or~\ref{line_min_beta_pair}, then~$v$ is prediction mandatory due to~\Cref{lemma_min_beta_kh_enforce_predmand}.
	Also, note that $\oj(Q) \geq |Q \setminus \OPT|$, since every vertex in~$Q$ is prediction mandatory at some point. If some $v \in Q$ is not in $\OPT$, then $v$ is not mandatory and, as argued above, has $\oj(v) \ge 1$.
	
	We divide the proof in three cases.
	For a pair $\{v, u\}$ as in Line~\ref{line_min_beta_pair}, note that, due to~\Cref{lemma_min_beta_kh_unique_loop} (and as argued before the proof), $u$~is not considered more than once, and is not considered in any of the previous cases.
	
	\begin{enumerate}
		\item If $|W| = 1$, then some vertex~$v$ was queried in Line~\ref{line_min_beta_query_first} because $\pred{w}_u$ enforces $v$, and $u$ is not queried by the algorithm due to Lemma~\ref{lemma_min_no_witness_implies_solved}.
		Then it suffices to note that $\{v, u\}$ is a witness set to see that $|Q \cup W| \leq |\OPT \cap (Q \cup \{u,v\})| + \oj(Q)$.
		\item Consider $|W| = 2$.
		If $W$ is a pair of the form $\{u, w\}$ queried in Line~\ref{line_min_beta_wit_trio}, then $\oj(v) \geq 1$ because~$\pred{w}_u$ enforces~$v$ but $w_u \notin I_v$ (as $v$ was not queried in Line~\ref{line_min_beta_mand_trio}).
		We can conceptually move this contribution in the hop distance to~$u$, making $\oj(v) := \oj(v) - 1$ and $\oj(u) := \oj(u) + 1$.
		If $v$ is considered another time in Line~\ref{lin_min_beta_cond_trio} or in another point of the analysis because it is enforced by some predicted weight, then it has to be the predicted weight of a vertex $u' \neq u$, so we are not counting the contribution to the hop distance more than once.
		If $W$ is a pair of the form $\{v,u\}$ queried in Line~\ref{line_min_beta_query_first} and in Line~\ref{line_min_beta_kh_mandatory_second} directly afterwards, then either $W \subseteq \OPT$ or $\oj(v) \ge 1$: It holds that~$u$ is mandatory, so if~$v$ is not in $\OPT$ then it suffices to see that $\pred{w}_u$ enforces $v$ to conclude that $\oj(v)\ge 1$.
		Either way, the fact that~$W$ is a witness set is enough to see that $|Q \cup W| \leq |\OPT \cap (Q \cup W)| + \oj(Q) + \oj(W)$.
		\item If $|W| = 3$, then $W = \{u,v,w\}$ as in Line~\ref{lin_min_beta_cond_trio}, and $|Q \cup W| = \gamma + 1$.
		As $v$ is queried in Line~\ref{line_min_beta_mand_trio}, it is mandatory by~\Cref{lemma_min_kh_enforce_mandatory}.
		Since $\{u,w\}$ is a witness set, it holds that~$v$ and at least one of $\{u, w\}$ are contained in any feasible solution.
		This implies that at least $\frac{\gamma}{\gamma+1} \cdot |Q \cup W| - \oj(Q)$ of the vertices in $Q \cup W$ are in $\OPT$, so $|Q \cup W| \leq (1 + \frac{1}{\gamma})(|\OPT \cap (Q \cup W)| + \oj(Q))$.
	\end{enumerate}
	
	The remaining vertices queried in Lines~\ref{line_min_beta_kh_mandatory_first} and~\ref{line_min_beta_kh_mandatory_second} are in any feasible solution and querying them never worsens the consistency. 
\end{proof}

\subsection{An error-sensitive algorithm w.r.t.\ the mandatory-query distance}

We continue by designing a learning-augmented algorithm with a linear error dependency on the mandatory query distance $k_M$. Before we do so, we provide the following lower bound stating that we can only hope for a slightly worse consistency and robustness tradeoff if we want a linear error dependency on $k_M$.

\begin{restatable}{theorem}{ThmLBMandQueryDist}
	\label{theo_lb_sym_diff}
	Let $\gamma \in \mathbb{R}_{\ge 2}$ be fixed. 
	If a deterministic algorithm for hypergraph orientation under explorable uncertainty with predictions is
	$\gamma$-robust,
	then it cannot have competitive ratio better than $1 + \frac{1}{\gamma-1}$ for $k_M = 0$.
	If an algorithm has competitive ratio $1 + \frac{1}{\gamma-1}$ for $k_M = 0$, then it cannot be better than
	$\gamma$-robust.
\end{restatable}

\begin{proof}
	We first establish the following auxiliary claim, which is slightly weaker than
	the statement of the theorem:
	
	\begin{claim}
		\label{claim:LBM}%
		Let $\gamma'\ge 2$ be a fixed rational number. Every deterministic algorithm for the hypergraph orientation
		problem has competitive ratio at least $1+\frac{1}{\gamma'-1}$ for $k_M=0$
		or has competitive ratio at least $\gamma'$ for arbitrary $k_M$.
	\end{claim}
	
	Let $\gamma' = \frac{a}{b}$, with integers $a \geq 2b > 0$.
	Consider an instance with vertices $\{1,\ldots,a\}$, hyperedges $S_i = \{i, b+1, b+2, \ldots, a\}$ for $i = 1, \ldots, b$, and intervals and predicted weights as depicted in Figure~\ref{fig_lb_sym_diff_input}.
	Suppose without loss of generality that the algorithm queries vertices $\{1,\ldots,b\}$ in the order $1, 2, \ldots, b$, and the vertices $\{b+1,\ldots,a\}$ in the order $b+1, b+2, \ldots, a$.
	Let the predictions be correct for $b+1, \ldots, a-1$, and $w_{1}, \ldots, w_{b-1} \notin I_{b+1} \cup \ldots \cup I_a$.

	\begin{figure*}[tb]
		\centering
		\captionsetup[subfigure]{justification=centering}
		\begin{subfigure}{0.25\textwidth}
			\begin{tikzpicture}[line width = 0.3mm, scale=0.8]
				\intervalpr{$I_1$}{0}{2}{0}{1.5}{0.5}
				\intervalpr{$I_2$}{0}{2}{0.5}{1.5}{0.5}
				\intervalp{$I_b$}{0}{2}{1.5}{1.5}
				\path (1, 0.5) -- (1, 1.6) node[font=\LARGE, midway, sloped]{$\dots$};
				
				\intervalpr{$I_{b+1}$}{1}{3}{2}{2.5}{2.5}
				\intervalpr{$I_{b+2}$}{1}{3}{2.5}{2.5}{2.5}
				\intervalp{$I_{a}$}{1}{3}{3.5}{2.5}
				\path (2, 2.5) -- (2, 3.6) node[font=\LARGE, midway, sloped]{$\dots$};
			\end{tikzpicture}
			\subcaption{}\label{fig_lb_sym_diff_input}
		\end{subfigure}
		\begin{subfigure}{0.25\textwidth}
			\begin{tikzpicture}[line width = 0.3mm, scale=0.8]
				\intervalpr{$I_1$}{0}{2}{0}{1.5}{0.5}
				\intervalpr{$I_2$}{0}{2}{0.5}{1.5}{0.5}
				\intervalpr{$I_b$}{0}{2}{1.5}{1.5}{0.5}
				\path (1, 0.5) -- (1, 1.6) node[font=\LARGE, midway, sloped]{$\dots$};
				
				\intervalpr{$I_{b+1}$}{1}{3}{2}{2.5}{2.5}
				\intervalpr{$I_{b+2}$}{1}{3}{2.5}{2.5}{2.5}
				\intervalpr{$I_{a}$}{1}{3}{3.5}{2.5}{1.5}
				\path (2, 2.5) -- (2, 3.6) node[font=\LARGE, midway, sloped]{$\dots$};
			\end{tikzpicture}
			\subcaption{}\label{fig_lb_sym_diff_afirst}
		\end{subfigure}
		\begin{subfigure}{0.25\textwidth}
			\begin{tikzpicture}[line width = 0.3mm, scale=0.8]
				\intervalpr{$I_1$}{0}{2}{0}{1.5}{0.5}
				\intervalpr{$I_2$}{0}{2}{0.5}{1.5}{0.5}
				\intervalpr{$I_b$}{0}{2}{1.5}{1.5}{1.5}
				\path (1, 0.5) -- (1, 1.6) node[font=\LARGE, midway, sloped]{$\dots$};
				
				\intervalpr{$I_{b+1}$}{1}{3}{2}{2.5}{2.5}
				\intervalpr{$I_{b+2}$}{1}{3}{2.5}{2.5}{2.5}
				\intervalpr{$I_{a}$}{1}{3}{3.5}{2.5}{2.5}
				\path (2, 2.5) -- (2, 3.6) node[font=\LARGE, midway, sloped]{$\dots$};
			\end{tikzpicture}
			\subcaption{}\label{fig_lb_sym_diff_bfirst}
		\end{subfigure}
				%
		\caption{Intervals, precise weights and predicted weights as used in the lower bound instance in the proof of~\Cref{theo_lb_sym_diff}.}
		\label{fig_lb_sym_difference}
	\end{figure*}
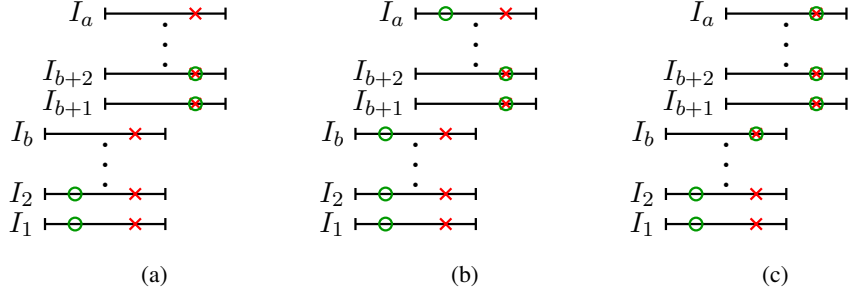
	
	If the algorithm queries vertex~$a$ before vertex~$b$, then the adversary sets $w_a \in I_{b}$ and $w_{b} \notin I_a$.
	(See Figure~\ref{fig_lb_sym_diff_afirst}.)
	This forces a query in all vertices $\{1,\ldots,b\}$ as they become mandatory by~\Cref{cor_min_left_mandatory}, so the algorithm queries all $a$ vertices, while the optimal solution queries only the~$b$ vertices $\{b+1,\ldots,a\}$.
	Thus the competitive ratio is at least $\frac{a}{b} = \jnew{\gamma'}$ if~$k_M$ can be arbitrarily large.
	
	If the algorithm queries~$b$ before~$a$, then the adversary sets $w_a = \pred{w}_a$ and $w_{b} \in I_a$; see Figure~\ref{fig_lb_sym_diff_bfirst}.
	This forces the algorithm to query all remaining vertices in $\{b+1,\ldots,a\}$ as they become mandatory by~\Cref{cor_min_left_mandatory}, i.e., $a$ queries in total, while the optimum queries only the $a-b$ vertices in $\{b+1,\ldots,a\}$. 
	Note, however, that $k_M = 0$, since the right-side vertices $\{b+1,\ldots,1\}$
	are mandatory for both predicted and precise weights by~\Cref{lema_mandatory_min}, while $\{1, \ldots, b\}$ are neither mandatory nor prediction mandatory.
	Thus, the competitive ratio is at least $\frac{a}{a-b} = 1 + \frac{1}{\jnew{\gamma'}-1}$ for $k_M = 0$.
	This concludes the proof of~\Cref{claim:LBM}.
	
	
	
	Now we are ready to prove the theorem.
	Let $\gamma\ge 2$ be a fixed rational. Assume that there is a deterministic algorithm
	that is $\gamma$-robust and has competitive ratio strictly smaller than
	$1 + \frac{1}{\gamma-1}$, say $1+\frac{1}{\gamma+\varepsilon-1}$ with $\varepsilon>0$,
	for $k_M = 0$. Let $\gamma'$ be a rational number with $\gamma < \gamma' <\gamma+\varepsilon$.
	Then the algorithm has competitive ratio strictly smaller than $\gamma'$ for arbitrary $k_M$
	and competitive ratio strictly smaller than $1+\frac{1}{\gamma'-1}$ for $k_M=0$,
	a contradiction to Claim~\ref{claim:LBM}. This shows the first statement of the theorem.
	
	Let $\gamma\ge 2$ again be a fixed rational. Assume that there is a deterministic algorithm
	that has competitive ratio $1 + \frac{1}{\gamma-1}$ for $k_M=0$ and is
	$(\gamma-\varepsilon)$-robust, where $\varepsilon>0$.
	As there is a lower bound of $2$ on the robustness of any deterministic algorithm, no such algorithm can exist for $\gamma=2$. So we only
	need to consider the case $\gamma>2$ and $\gamma-\varepsilon\ge 2$. Let $\gamma'$ be a rational number
	with $\gamma-\varepsilon<\gamma'<\gamma$. Then the algorithm has competitive ratio strictly smaller than
	$1 + \frac{1}{\gamma'-1}$ for $k_M=0$ and competitive ratio strictly smaller than
	$\gamma'$ for arbitrary $k_M$, a contradiction to Claim~\ref{claim:LBM}. This shows
	the second statement of the theorem.
\end{proof}

This lower bound shows that the $k_h$-dependent guarantee of Algorithm~\ref{ALG_min_beta} (cf.~\Cref{thm_minimum:hop}) cannot translate to the error measure $k_M$. Instead, for any $\gamma \in \mathbb{N}_{\ge 2}$, we aim for a $(1+\frac{1}{\gamma-1})$-consistent and $\gamma$-robust algorithm with linear error dependency on $k_M$.

To that end, we prove the following tight bound
by presenting the new Algorithm~\ref{ALG_min_alpha} with dependency on $k_M$.
We remark that this algorithm only uses the initial set of prediction mandatory vertices, and otherwise ignores the predicted weights.
Since access to this set is sufficient to execute the algorithm, it requires strictly less predicted information than the Algorithm~\ref{ALG_min_alpha} which relies on having access to the actual predicted weights.

\begin{restatable}{theorem}{ThmMinAlpha}
	\label{thm:min-alpha}
	There is an algorithm for hypergraph orientation under explorable uncertainty with predictions that, given an integer parameter $\gamma \geq 2$, 
	has a competitive ratio of $\min\{ (1+\frac{1}{\gamma-1}) \cdot (1 + \frac{k_M}{\opt}), \gamma\}$. 
\end{restatable}

\subsubsection*{Overview of the algorithmic ideas}

Before we give a formal proof of the theorem, we start by sketching the main ideas.
The algorithm emulates the two-stage structure of the offline algorithm. 
Recall that the offline algorithm in a first stage queries all mandatory vertices and in a second stage queries a minimum vertex cover in the remaining vertex cover instance.
Since blindly following the offline algorithm based on the predicted weights would lead to a competitive ratio of $n$, the algorithm augments both stages with additional queries. 
Algorithm~\ref{ALG_min_alpha} implements the augmented first stage in Lines~\ref{line_min_param_cond} to~\ref{line_min_param_small} and afterwards executes the second stage.

To start the first phase, the algorithm computes the set~$P$ of initial prediction mandatory vertices (Lemma~\ref{lema_mandatory_min}).
In contrast to the $k_h$-dependent algorithm, we fix the set $P$ and do \emph{not} recompute it when the instance changes.
Then the algorithm tries to find a vertex $p \in P$ that is part of a witness set $\{p, b\}$.
If~$|P|\ge\gamma-1$, we query a set $P' \subseteq P$ of size $\gamma-1$ that includes~$p$, plus~$b$ (we allow $b \in P'$).
This is clearly a witness set of size at most $\gamma$, which ensures that the queries do not violate the $\gamma$-robustness.
Also, at least a $\frac{\gamma-1}{\gamma}$ fraction of the queried vertices are in~$P$, and every vertex in $P \setminus \OPT$ is in $\mathcal{I}_P \setminus \mathcal{I}_R$ and, thus, contributes to the mandatory query distance $k_M$.
This ensures, at least locally, that the queried vertices do not violate the error-dependent consistency.
We then repeatedly query known mandatory vertices, remove 
them from~$P$ and repeat 
without recomputing~$P$, until~$P$ is empty or no vertex in~$P$ is part of a witness set.

We may have one last iteration of the loop where $|P| < \gamma -1$.
After that, the algorithm will proceed to the second phase, querying a minimum vertex cover of the current vertex cover instance and vertices that become known mandatory by~\Cref{cor_min_left_mandatory}.
For the second phase itself, we can use that a minimum vertex cover of the vertex cover instance (cf.~\Cref{def:vertex_cover_instance}) is a lower bound on the optimal solution for the remaining instance by~\Cref{lem:witness}. Since all queries of Line~\ref{p2-line3} are mandatory, the queries of the Lines~\ref{p2-line2} and~\ref{p2-line3} are $2$-robust for the remaining instance. Even in combination with the additional at most $\gamma -2$ queries of the last iteration of the loop, this is still $\gamma$-robust.
It is not hard to show that each query of Line~\ref{p2-line3} contributes an error to $k_M$, which completes the argument.

\begin{algorithm}[tb]
	\KwIn{Hypergraph $H=(V,E)$, intervals $I_v$ and predictions $\w_v \in I_v$ for all $v \in V$}
	$P \leftarrow$ set of initial prediction mandatory vertices (characterized in Lemma~\ref{lema_mandatory_min})\;
	\While{$\exists p \in P$ and an unqueried vertex $b$ where $\{p, b\}$ is a witness set for the current instance by \Cref{lem:witness} \label{line_min_param_cond}}{
		\uIf{$|P| \geq \gamma-1$ \label{line_min_param_size}}{
			pick $P' \subseteq P$ with $p \in P'$ and $|P'| = \gamma-1$\;
			query $P' \cup \{b\}$, $P \leftarrow P \setminus (P' \cup \{b\})$\; \label{line_min_param_big}
			\lWhile{there is a known mandatory vertex~$v$ by~\Cref{cor_min_left_mandatory}}{query $v$, $P \leftarrow P \setminus \{v\}$ \label{line_min_param_mandatory}}
		} \lElse{query $P$, $P \leftarrow \emptyset$ \label{line_min_param_small}}
	}
	Compute and query a minimum vertex cover~$Q'$ on the current vertex cover instance\label{p2-line2}\; 
	\lWhile{there is a known mandatory vertex~$v$ by~\Cref{cor_min_left_mandatory}}{query $v$\label{p2-line3}}
	\caption{Algorithm for hypergraph orientation without prediction mandatory vertices}
	\caption{Algorithm for hypergraph orientation under uncertainty w.r.t.\ error measure~$k_M$}
	\label{ALG_min_alpha}
\end{algorithm}

\subsubsection*{Precise algorithm and formal analysis}

We proceed by turning these arguments into a formal analysis of Algorithm~\ref{ALG_min_alpha} to prove~\Cref{thm:min-alpha}. To that end, we first show the following auxiliary lemma. Recall that $\mathcal{I}_P$ denotes the set of (initially) prediction mandatory vertices for the instance and $\mathcal{I}_R$ denotes the set of vertices that are mandatory for the precise weights.

\begin{restatable}{lemma}{LemmaMinNoPredMandAfterVC}
	\label{lemma_min_no_pred_mand_after_vc}
	Every vertex queried in Line~\ref{p2-line3} of Algorithm~\ref{ALG_min_alpha} is in $\mathcal{I}_R \setminus \mathcal{I}_P$, i.e., mandatory but not (initially) prediction mandatory.
\end{restatable}

\begin{proof}
	Clearly every such vertex is in~$\mathcal{I}_R$  because it is known {to be} mandatory by~\Cref{cor_min_left_mandatory}, so it remains to prove that it is not in $\mathcal{I}_P$.
	Consider a hyperedge~$S$.
	If a vertex $u \in S \cap \mathcal{I}_P$ is not queried in Line~\ref{p2-line2}, then the condition for identifying a witness set in Line~\ref{line_min_param_cond} of Algorithm~\ref{ALG_min_alpha} implies that, before Line~\ref{p2-line2} is executed, $u$~is not leftmost in~$S$ and does not intersect the leftmost vertex in~$S$.
	Thus, $v$ cannot become known mandatory in Line~\ref{p2-line3} and, therefore, is not queried.
	This implies that a vertex queried in Line~\ref{p2-line3} cannot be contained in $\mathcal{I}_P$ and, thus, the lemma.
\end{proof}

Using the auxiliary lemma, we now proceed to prove~\Cref{thm:min-alpha}.

\begin{proof}[Proof of~\Cref{thm:min-alpha}]
	Before we prove the performance bounds, we remark that the algorithm clearly solves the given instance by definition of the final two lines of the algorithm and~\Cref{lem:VC-property}.
	Next, we separately show that the algorithm executes at most $\gamma \cdot \opt$ queries and at most $(1 + \frac{1}{\gamma-1})(1 + \frac{k_M}{\opt}) \cdot \opt$ queries.
	
	\paragraph{Proof of $|\ALG| \le \gamma \cdot \opt$ (robustness).}
	Given $P' \cup \{b\}$ queried in Line~\ref{line_min_param_big}, at least one vertex is in any feasible solution since $\{b, p\}$ is a witness set and, thus, $P' \cup \{b\}$ is a witness set of size $\gamma$. Therefore, querying $P' \cup \{b\}$ never worsens the robustness below $\gamma$.
	
	Line~\ref{line_min_param_small} is executed at most once, since the size of~$P$ never increases.
	Fix an optimum solution $\OPT$, and let~$V'$ be the set of unqueried vertices before Line~\ref{line_min_param_small} is executed (or before Line~\ref{p2-line2} if Line~\ref{line_min_param_small} is never executed).
	Let~$P$ be the set of vertices queried in Line~\ref{line_min_param_small}, and let $Q$ be the queries in Line~\ref{p2-line2}.
	
	If the orientation is already known before querying $P$ and $Q$, then it must hold $P=Q=\emptyset$ and the lemma clearly holds.
	If the orientation is not yet known at this point, then $|\OPT \cap V'| \geq 1$, so $|P| \leq \gamma - 2$ implies $|P| \leq (\gamma - 2) \cdot |\OPT \cap V'|$.
	Also, since~$Q$ is a minimum vertex cover of the vertex cover instance, we get $|Q| \leq |\OPT \cap V'|$ by~\Cref{lem:witness}.
	Let~$M$ be the set of vertices in~$V'$ that are queried in Line~\ref{p2-line3}; clearly $M \subseteq \OPT \cap V'$ as all those vertices are mandatory.
	Thus $|P| + |Q| + |M| \leq \gamma \cdot |\OPT \cap V'|$.
	
	The vertices queried in Line~\ref{line_min_param_mandatory} are in any feasible solution, which implies the robustness bound.

	\paragraph{Proof of $|\ALG| \le (1+\frac{1}{\gamma-1}) \cdot (1 + \frac{k_M}{\opt}) \cdot \opt$ (consistency and error-dependency).}
	Fix an optimum solution $\OPT$.
	In the following, we will show for various disjoint subsets $J\subseteq V$ that $|J \cap \ALG| \le (1+\frac{1}{\gamma-1})\cdot (|\OPT \cap J| + k_{J})$, where $k_{J} \leq |J \cap (\mathcal{I}_P \sym \mathcal{I}_R)|$.
	The union of the subsets $J$ will contain $\ALG$, so it is clear that the bound of $(1+\frac{1}{\gamma-1}) \cdot (1 + \frac{k_M}{\opt})$ on the competitive ratio of the algorithm follows.
	
	Vertices queried in Lines~\ref{line_min_param_mandatory} of Algorithm~\ref{ALG_min_alpha} are part of any feasible solution, hence the set $P_0$ of these vertices satisfies $|P_0|\le | \OPT \cap P_0|$.
	
	Given $P' \cup \{b\}$ queried in Line~\ref{line_min_param_big}, at least $\frac{\gamma-1}{\gamma}$ of the vertices in $P' \cup \{b\}$ are prediction mandatory for the initial instance by choice of $P'$.
	Among those, let~$k' \leq k_M$ be the number of vertices in $\mathcal{I}_P \setminus \mathcal{I}_R$.
	Then, $|\OPT \cap (P' \cup \{b\})| \geq \frac{\gamma-1}{\gamma} \cdot |P' \cup \{b\}| - k'$, which gives the desired bound, i.e., $|P' \cup \{b\}| \le (1+\frac{1}{\gamma-1})\cdot (|\OPT \cap (P' \cup \{b\})| + k')$.
	
	Every vertex queried in Line~\ref{line_min_param_small} that is not in $\OPT$ is in $\mathcal{I}_P \setminus \mathcal{I}_R$.
	Hence, if there are $k''$ such intervals, then the set $P$ of vertices queried in Line~\ref{line_min_param_small} satisfies $|P| \le |\OPT \cap P| + k'' < (1+\frac{1}{\gamma})\cdot (|\OPT \cap P| + k'')$.
	
	Let~$V'$ be the set of unqueried vertices before Line~\ref{p2-line2} is executed, and let $Q$ be the queries in Line~\ref{p2-line2}.
	Then $|Q| \leq |\OPT \cap V'|$ because~$Q$ is a minimum vertex cover of the vertex cover instance, which is a lower bound on $\opt$ by~\Cref{lem:witness}.
	Let~$M$ be the set of vertices that are queried in Line~\ref{p2-line3}.
	It holds that $|Q \cup M| \leq |\OPT \cap V'| + |M|$, so the claimed bound follows from Lemma~\ref{lemma_min_no_pred_mand_after_vc}.
\end{proof}

\subsection{Non-integral parameter $\gamma$ via randomization}

The parameter $\gamma$ in~\Cref{thm_minimum:hop,thm:min-alpha} is restricted to integral values since the corresponding algorithms use it to determine sizes of query sets. Nevertheless, a generalization to arbitrary $\gamma \in \RR_{+}$ is possible at a small loss in the guarantee. We give the following rigorous upper bound on the achievable tradeoff between consistency and robustness with linear error-dependency on $k_M$. 

The following two theorems are shown via an adapted version of the corresponding proof given in~\cite{EdLMS22}.

\begin{restatable}{theorem}{ThmMinFractionalGamma}
	\label{thm:min-problem-arbitrary-gamma}
	For any real number $\gamma \geq 2$, there is a randomized algorithm for the hypergraph orientation problem under explorable uncertainty with predictions that achieves a competitive ratio of $\min\{(1+\frac{1}{\gamma-1}+\xi)\cdot(1+\frac{k_M}{\opt}), \gamma\}$, for $\xi \leq \frac{\gamma - \lfloor \gamma \rfloor}{(\gamma-1)^2} \leq1.$
\end{restatable}

\begin{proof}
	For $\gamma \in \ZZ$, we run Algorithm \ref{ALG_min_alpha} and achieve the performance guarantee from Theorem \ref{thm:min-alpha}. 
	Assume $\gamma \notin\ZZ$, and let $\{\gamma\}:=\gamma - \lfloor \gamma \rfloor = \gamma - \lceil \gamma \rceil +1$ denote its fractional part. We run the following randomized variant of Algorithm \ref{ALG_min_alpha}. We randomly chose $\gamma'$ as $\lceil \gamma \rceil$  with probability $\{\gamma\}$ and as $\lfloor \gamma \rfloor$ with probability $1-\{\gamma\}$, and then we run the algorithm with $\gamma'$ instead of $\gamma$. We show that the guarantee from Theorem \ref{thm:min-alpha} holds in expectation with an additive term less than $\{\gamma\}$, more precisely, we show the competitive ratio
	\[
	\min\left\{\left( 1+\frac{1}{\gamma-1} +\xi \right) \cdot \left(1+\frac{k_M}{\opt} \right), \gamma\right\}, \text{ for } \xi = \frac{ \{\gamma\}(1-\{\gamma\}) }{(\gamma -1) \lfloor \gamma \rfloor (\lfloor \gamma \rfloor - 1)} \leq \frac{\{\gamma\}}{(\gamma-1)^2} .
	\]
	
	Following the arguments in the proof of Theorem \ref{thm:min-alpha} on the robustness, the ratio of the algorithm's number of queries $|\ALG|$ and $|\OPT|$ is bounded by $\gamma'$. In expectation the robustness is
	\begin{align*}
		\EX\left[{\gamma'}\right] &= (1-\{\gamma\}) \cdot \lfloor \gamma \rfloor + \{\gamma\} \cdot \lceil \gamma \rceil \\
		&= (1-\{\gamma\}) \cdot (\gamma - \{\gamma\}) + \{\gamma\} \cdot (\gamma - \{\gamma\} +1) \\
		&= \gamma.
	\end{align*}
	
	The error-dependent bound on the competitive ratio is in expectation (with $\opt$ and $k_M$ not being random variables)
	\begin{align*}
		\EX\left[\left(1+\frac{1}{\gamma'-1}\right) \cdot \left( 1+ \frac{k_M}{\opt} \right)\right] = \left(1+\EX\left[\frac{1}{\gamma'-1}\right]\right)\cdot \left( 1+ \frac{k_M}{\opt} \right).
	\end{align*}
	
	Applying simple algebraic transformations, we obtain
	\begin{align*}
		\EX\left[\frac{1}{\gamma'-1}\right] &= \frac{1-\{\gamma\}}{\lfloor \gamma \rfloor-1} + \frac{\{\gamma\}}{\lceil \gamma \rceil-1}\ 
		=\ \frac{1-\{\gamma\}}{\gamma - \{\gamma\}-1} + \frac{\{\gamma\}}{\gamma - \{\gamma\}}\\
		&= \frac{(1-\{\gamma\})(\gamma - \{\gamma\}) + \{\gamma\}(\gamma - \{\gamma\}-1)}{(\gamma - \{\gamma\}-1)(\gamma - \{\gamma\})} \\
		&= \frac{ \gamma -2 \{\gamma\} }{(\gamma - \{\gamma\}-1)(\gamma - \{\gamma\})} \ 
		=\ \frac{1}{\gamma-1} - \frac{1}{\gamma-1} + \frac{ \gamma -2\{\gamma\}}{(\gamma - \{\gamma\}-1)(\gamma - \{\gamma\})} \\
		&=\ \frac{1}{\gamma-1} + \frac{ \{\gamma\}(1-\{\gamma\}) }{(\gamma-1) (\gamma - \{\gamma\}-1)(\gamma - \{\gamma\})} \ = \ \frac{1}{\gamma-1} + \frac{ \{\gamma\}(1-\{\gamma\}) }{(\gamma-1) \lfloor \gamma \rfloor (\lfloor \gamma \rfloor - 1)}. 
	\end{align*}
	
	Hence, the error-dependent bound on the competitive ratio is in expectation
	\[
	\left( 1+\frac{1}{\gamma-1} +\xi \right) \cdot \left(1+\frac{k_M}{\opt} \right) \text{ with } \xi = \frac{ \{\gamma\}(1-\{\gamma\}) }{(\gamma -1) \lfloor \gamma \rfloor (\lfloor \gamma \rfloor -1)} \leq \frac{\{\gamma\}}{(\gamma-1)^2} \le 1,
	\]
	which concludes the proof.
\end{proof}

We can repeat essentially the same proof to obtain the analogous result for Algorithm~\ref{ALG_min_beta} and linear error-dependency on $k_h$.

\begin{restatable}{theorem}{ThmMinFractionalGammaHop}
	\label{thm:min-problem-arbitrary-gamma-hop}
	For any real number $\gamma \geq 2$, there is a randomized algorithm for the hypergraph orientation problem under explorable uncertainty with predictions that achieves a competitive ratio of $\min\{(1+\frac{1}{\gamma}+\xi)\cdot(1+\frac{k_h}{\opt}), \gamma\}$, for $\xi \leq \frac{1}{48} < 0.021$.
\end{restatable}

\begin{proof}
	\newcommand{\gfloor}{\lfloor \gamma \rfloor}
	\newcommand{\gceil}{\lceil \gamma \rceil}
	For $\gamma \in \ZZ$, we run Algorithm \ref{ALG_min_beta} and achieve the performance guarantee from Theorem \ref{thm_minimum:hop}. 
	Assume $\gamma \notin\ZZ$. As before, let $\{\gamma\}:=\gamma - \lfloor \gamma \rfloor = \gamma - \lceil \gamma \rceil +1$ denote its fractional part. We run the following randomized variant of Algorithm \ref{ALG_min_beta}. We randomly chose $\gamma'$ as $\lceil \gamma \rceil$  with probability $\{\gamma\}$ and as $\lfloor \gamma \rfloor$ with probability $1-\{\gamma\}$, and then we run the algorithm with $\gamma'$ instead of $\gamma$. We show that the guarantee from Theorem \ref{thm_minimum:hop} holds in expectation with an additive term less than $\{\gamma\}$, more precisely, we show the competitive ratio
	\[
	\min\left\{\left( 1+\frac{1}{\gamma} +\xi \right) \cdot \left(1+\frac{k_h}{\opt} \right), \gamma\right\}, \text{ for } \xi = \frac{ \{\gamma\}(1-\{\gamma\}) }{(\gamma -1) \lfloor \gamma \rfloor (\lfloor \gamma \rfloor - 1)} \leq \frac{\{\gamma\}}{(\gamma-1)^2} .
	\]
	
	Exactly as in the proof of~\Cref{thm:min-problem-arbitrary-gamma}, we get
	$
	\EX[\gamma'] =  \gamma.
	$
	The error-dependent bound on the competitive ratio is in expectation (with $\opt$ and $k_h$ not being random variables)
	\begin{align*}
		\EX\left[\left(1+\frac{1}{\gamma'}\right) \cdot \left( 1+ \frac{k_h}{\opt} \right)\right] = \left(1+\EX\left[\frac{1}{\gamma'}\right]\right)\cdot \left( 1+ \frac{k_M}{\opt} \right).
	\end{align*}
	
	Applying simple algebraic transformations, we obtain
	\begin{align*}
		\EX\left[\frac{1}{\gamma'}\right] &= \{\gamma\}\frac{1}{\gceil} + (1-\{\gamma\})\cdot\frac{1}{\gfloor}\\
		&= \{\gamma\}\frac{1}{\gceil} + (1-\{\gamma\})\cdot\frac{1}{\gfloor} \ -\  \frac{1}{\gamma} \ +\  \frac{1}{\gamma}\\
		&= \frac{1}{\gamma} + \frac{1}{\gamma \gceil \gfloor}\cdot \bigg( \{\gamma\} \gamma \gfloor  + (1- \{\gamma\}) \cdot \gceil  \gamma - \gceil \gfloor  \bigg).
	\end{align*}
	As we consider fractional $\gamma >2$, it holds $\lceil \gamma\rceil = \gfloor + 1$. Using this, we rewrite the term in brackets as
	\begin{align*}
		& \{\gamma\} \gamma \gfloor  + (1- \{\gamma\}) \cdot (\gfloor +1)  \gamma - (\gfloor +1)  \gfloor\\
		& = \{\gamma\} \gamma \gfloor  + (\gfloor +1)  \gamma - \{\gamma\} \cdot (\gfloor +1)  \gamma - (\gfloor +1)  \gfloor\\
		& = (\gfloor +1)(\gamma -\gfloor) - \{\gamma\} \cdot \gamma\\
		& = \{\gamma\} (\gfloor + 1 - \gamma) \\
		& = \{\gamma\}(1-\{\gamma\}),
	\end{align*}
	where the third and fourth equalities come from the fact that $\gamma -\gfloor = \{\gamma\}$. Note that for $\{\gamma\} \geq 0$, the expression $\{\gamma\}(1-\{\gamma\})$ is at most $1/4$, where it reaches its maximum for $\{\gamma\} =1/2$. Further, notice that for fractional $\gamma > 2$ it holds that $\gamma \gceil \gfloor \geq 12$. We conclude that
	\[
	\EX{\frac{1}{\gamma'}} \leq \frac{1}{\gamma} + \frac{ \{\gamma\}(1-\{\gamma\}) }{\gamma \gceil \gfloor} \leq \frac{1}{\gamma} + \frac{1}{48} < \frac{1}{\gamma} + 0.021
	\]
	which proves the theorem.
\end{proof}

\section{Sorting under uncertainty} 
\label{sec:sorting}

 In this section, we consider the special case of the hypergraph orientation problem, where  {the input graph is a simple graph $G=(V,E)$ that satisfies $\{u,v\} \in E$ if and only if $I_v \cap I_u \not= \emptyset$. That is, $G$ corresponds to the interval graph induced by the uncertainty intervals $\mathcal{I}= \{I_v \mid v\in V\}$.}
 To orient such a graph, we have to, for each pair of intersecting intervals, decide which one has the smaller precise weight.
 An orientation of the graph defines an order of the intervals according to their precise weights (and vice versa).
 Thus, the problem corresponds to the problem of sorting a single set of uncertainty intervals.
 Note that, by querying vertices, the uncertainty intervals change and, thus, the graph induced by the intervals also changes. When we speak of the \emph{current interval graph}, we refer to the interval graph induced by the uncertainty intervals \emph{after} all previous queries.
 
 As the main result of this section, we give a learning-augmented algorithm for sorting under explorable with predictions that is $1$-consistent and $2$-robust with a linear error-dependency for any $k \in \{k_{\#}.k_M,k_h\}$. Clearly, no algorithm can be better than $1$-consistent, and no deterministic algorithm can be better than $2$-robust according to the adversarial lower bound.
 Before we give our algorithmic results, we show that a sublinear error dependency is not possibly by giving a lower bound example.
 
 \begin{restatable}{theorem}{ThmLBAllErrorMeasures}
 	\label{thm_lb_error_measure}
 	Any deterministic algorithm for sorting or hypergraph orientation under explorable uncertainty with predictions (even for pairwise disjoint hyperedges) has a competitive ratio $\rho\geq \min\{1+\frac{k}{\opt},2\}$, for any error measure $k \in \{k_{\#}, k_M, k_h\}$.
 \end{restatable}
 
 \begin{proof}
 	Consider
 	the input instance of the hypergraph orientation problem consisting of a single edge $\{v,u\}$ with intervals and predicted weights as
 	shown in Figure~\ref{fig_lb_error_measure}.
 	
 	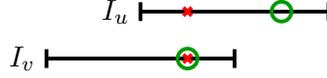
\begin{figure}[t]
 		\centering
 		\begin{tikzpicture}[thick,line width = \lw, scale=1.25]
 			\intervalpr{$I_v$}{0}{2}{0}{1.5}{1.5}
 			\intervalpr{$I_u$}{1}{3}{0.5}{1.5}{2.5}
 		\end{tikzpicture}
 		\caption{
 			Uncertainty intervals as well as the predicted and precise weights as used in the proof of~\Cref{thm_lb_error_measure}.
 		}
 		\label{fig_lb_error_measure}
 	\end{figure}
 	
 	If the algorithm starts querying~$I_v$, then the adversary sets $w_v = \pred{w}_v$ and the algorithm is forced to query~$u$. Then $w_u \in I_u \setminus I_v$, so the optimum queries only~$u$.
 	It is easy to see that $k_{\#} = k_M = k_h = 1$.
 	A symmetric argument holds if the algorithm starts querying~$u$. In that case, $w_u = \pred{w}_u$ which enforces to query $v$ with $w_v \in I_v \setminus I_u$. Taking multiple copies of this instance gives the  result for any $k \leq \opt$.
 \end{proof}
 
 \subsection{A Learning-augmented Algorithm for Sorting}
 
 As a main result of this section, we show the following upper bound that matches~\Cref{thm_lb_error_measure}.
 
 \begin{restatable}{theorem}{sortingsingleset}
 	\label{thm:sorting:singleset}
 	There exists a single polynomial-time algorithm for 
 	{sorting} 
 	under uncertainty with predictions that is $\min\{1 + \frac{k}{\opt}, 2\}$-competitive for any $k \in \{k_{\#},k_M,k_h\}$. 
 \end{restatable}
 
 The key observation that allows us to achieve improved results for orienting interval graphs is the simple characterization of mandatory vertices: any vertex with an interval that contains the precise weight of another vertex is mandatory~\cite{halldorsson19sortingqueries}. 
 This observation is a direct consequence of Lemma~\ref{lema_mandatory_min} and the structure of interval graphs.
 Analogously, each vertex with an interval that contains the predicted weight of another vertex is prediction mandatory. 
 Furthermore,  {Lemma~\ref{lem:witness}} 
 implies that any two vertices with intersecting intervals constitute a witness set.
 
 To obtain a guarantee of $\opt + k$ for any measure~$k$, our algorithm (cf.\ Algorithm~\ref{fig:sorting1cons2rob}) must trust the predictions as much as possible.
 That is, the algorithm must behave very close to the offline algorithm under the assumption that the predictions are correct. 
 Recall that the offline algorithm in a first stage queries all mandatory vertices and in a second stage queries a minimum vertex cover in the remaining vertex cover instance after the first stage queries.
 Algorithm~\ref{fig:sorting1cons2rob} emulates again these two stages. In contrast to the algorithms for general hypergraph orientation, we cannot afford to augment the stages with additional queries as we aim at achieving $1$-consistency. Thus, we need a new algorithm and cannot apply existing results.
 
 In the emulated first phase, our algorithm queries all prediction mandatory intervals (cf.\ Line~\ref{line:containpredicted}) and all intervals that are mandatory based on the already obtained information (cf.\ Lines~\ref{line:mandatoryproper} and~\ref{line:mandatory}).
 This phase clearly does not violate the $\opt + k$ guarantee for $k \in \{k_M,k_h\}$, as all queried known mandatory vertices (cf.\ Lines~\ref{line:mandatoryproper} and~\ref{line:mandatory}) are contained in $\OPT$ and all queried prediction mandatory vertices (cf.\ Line~\ref{line:containpredicted}) are either in $\OPT$ or contribute one to $k_M \le k_h$.
 We will show that the same holds for $k=k_{\#}$.
 However, the main challenge is to guarantee $2$-robustness. 
 Our key ingredient for ensuring 
 this is the following lemma, which we show in~\Cref{predorient:sec:cliquepartition}.
 
 \begin{algorithm}[th]
 	\KwIn{Interval graph $G=(V,E)$, intervals $I_v$ and predictions $\w_v$ for all $v \in V$}
 	$\mathcal{I}_P \gets$ set of prediction mandatory vertices\label{line:computepm}\;
 	\lWhile{there is a known mandatory vertex~$v$ by~\Cref{cor_min_left_mandatory}}{query $v$\label{line:mandatoryproper}}
 	$M_1 \gets$ vertices queried in Line~\ref{line:mandatoryproper}; $\mathcal{S} \gets \mathcal{I}_P \setminus M_1$ \label{line:fixS}\;
 		Query $\mathcal{S}$\label{line:containpredicted}\;
 	\lWhile{there is a known mandatory vertex~$v$ by~\Cref{cor_min_left_mandatory}}{query $v$\label{line:mandatory}}
 	$M_2 \gets$ set of vertices queried in Line~\ref{line:mandatory}\;
 	$\mathcal{C} \gets$ Clique partition of $\mathcal{S} \cup M_1 \cup M_2$ such that all isolated vertices $v$ satisfy either $v \in M_1 \cup M_2$ or $I_u \cap I_v \not= \emptyset$
 	for a  distinct $u \not\in \mathcal{S} \cup M_1 \cup M_2$ 
 	(computed using Lemma~\ref{lem:clique_partition_summary})\label{line:cliquepartition}\;
 	\While{the problem is unsolved\label{line:vcbeginsort}}{
 		\KwLet $P = x_1 x_2 \cdots x_p$ be a path component of the current interval graph with $p \geq 2$ in direction of non-increasing lower limits $L_{x_i}$\; \label{line:looppaths}
 		\lIf{$p$ is odd}{query $\{x_2, x_4, \ldots, x_{p-1}\}$\label{line:pathodd}}
 		\Else{
 			\lIf{
 				$x_1$ is the distinct partner of a critical isolated vertex $v$~($I_{x_1} \cap I_v \not= \emptyset$ and $v \not\in M_1 \cup M_2$; cf.~\Cref{lem:clique_partition_summary})
 			}{query $\{x_1,x_3,\ldots,x_{p-1}\}$\label{line:patheven1}}
 			\lElse{query $\{x_2, x_4, \ldots, x_{p}\}$\label{line:patheven2}}
 		}
 		\lWhile{there is a known mandatory vertex~$v$ by~\Cref{cor_min_left_mandatory}}{query $v$\label{line:mandatory2}}
 	}
 	\caption{Learning-augmented algorithm for sorting under explorable uncertainty with predictions} 
 \label{fig:sorting1cons2rob}
\end{algorithm}

\begin{restatable}{lemma}{LemClique}
\label{lem:clique_partition_summary}
For an instance of the 
{sorting} problem, let $\mathcal{I}_P$ be the set of prediction mandatory vertices and 
$M$ 
the set of known mandatory vertices after querying $\mathcal{I}_P$ (by exhaustively applying Corollary~\ref{cor_min_left_mandatory}).
Then, we can partition $\mathcal{I}_P \cup M$  into a set of disjoint cliques $\mathcal{C}$ such that each $v$ with $\{v\} \in \mathcal{C}$ either satisfies $v \in M$ or $I_v \cap I_u \not= \emptyset$ for a distinct $u \not\in \mathcal{I}_P \cup M$. 
The partition can be computed in polynomial time.
\end{restatable}

We can apply the lemma to the queries of the first phase of the algorithm (cf.\ Line~\ref{line:cliquepartition}).
Given the partition $\mathcal{C}$ of the lemma, we know that queries to vertices $v$ that are part of some $C \in \mathcal{C}$ with $|C| \ge 2$ (or mandatory, i.e., $v \in M$) do not violate the $2$-robustness as even the optimal solution can avoid at most one query per clique~\cite{halldorsson19sortingqueries}.
Thus, we only have to worry about vertices $v \not\in M$ with $\{v\} \in \mathcal{C}$.
We call such vertices \emph{critical isolated} vertices. 
But even for  {critical} isolated vertices $v$, the lemma gives us a distinct not yet queried $u$ with $I_v \cap I_u \not= \emptyset$, i.e., $\{v,u\}$ is a witness set.

In line with the offline algorithm, the second phase of the algorithm (cf.~Lines~\ref{line:vcbeginsort} to~\ref{line:mandatory2}) queries a minimum vertex cover of the remaining instance (the interval graph defined by the intervals of non-queried vertices).  
However, to guarantee $2$-robustness, we have to take the witness sets of the  {critical} isolated vertices into account when deciding which vertex cover to query.
To see this, consider the example of Figure~\ref{fig:ex:vcselection}:  
only $v_1$ is prediction mandatory, so the first phase of the algorithm just queries $v_1$.
After querying $v_1$, there are no prediction mandatory (or mandatory) vertices left.
As the only vertex queried in the first phase, $\{v_1\}$ must be part of every clique partition. Since $v_1$ is not mandatory, it would qualify as a critical isolated vertex, and $v_2$ is the only possible distinct partner of $v_1$ with an intersecting interval that~\Cref{lem:clique_partition_summary} could assign to $v_1$.
After querying $v_1$, the remaining vertex cover instance is the path $v_2,v_3,v_4,v_5$.
One possible minimum vertex cover of this instance is $\{v_3,v_5\}$, but querying this vertex cover renders $v_2$ and  {$v_4$} mandatory by~\Cref{cor_min_left_mandatory}.
Thus, the algorithm would query all five intervals, which violates the $2$-robustness as the optimal solution just queries $\{v_2,v_4\}$.
The example illustrates that the selection of the minimum vertex cover in the second phase is important to ensure $2$-robustness. 

\begin{figure}[t]
\centering
\begin{tikzpicture}[line width = \lw, scale = 1, transform shape]
	\intervalpr{$I_{v_1}$}{0}{3}{0}{0.5}{1}		
	\intervalpr{$I_{v_2}$}{2}{5}{0.7}{2.5}{3.5}		
	\intervalpr{$I_{v_3}$}{4}{7}{0}{5.5}{4.5}					
	\intervalpr{$I_{v_4}$}{6}{9}{0.7}{7.5}{7.5}	
	\intervalpr{$I_{v_5}$}{8}{11}{0}{9.5}{8.5}	
\end{tikzpicture}
\caption{Example showing that it is necessary to query a specific vertex cover in the second phase to ensure $2$-robustness. Circles illustrate precise weights and crosses illustrate the predicted weights.}
\label{fig:ex:vcselection}
\end{figure}
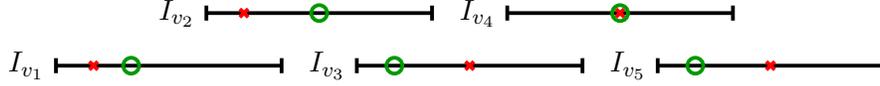

The next lemma shows that instances occurring in the second phase indeed have a structure similar to the example by exploiting that such instances have no prediction mandatory vertices.

\begin{restatable}{lemma}{PathLemma}
\label{lem:path}
Each connected component of an interval graph without prediction mandatory vertices is either a path or a single vertex.
\end{restatable}

\begin{proof}
First, observe that there are no intervals $I_v,I_u$ with $I_u \subseteq I_v$ as this would imply $\w_u \in I_v$, which contradicts the assumption as $v$ would be prediction mandatory.
Thus, the graph is  {a} proper  {interval graph}.
We claim that the graph contains no triangles; for proper interval graphs, this implies that each connected component is a path, because the 4-star $K_{1,3}$ is a forbidden induced subgraph~\cite{wegner67properinterval}.
Suppose 
there is a triangle $abc$, and assume that $L_a \leq L_b \leq L_c$; it holds that $U_a \leq U_b \leq U_c$ because no interval is contained in another.
Since~$I_a$ and~$I_c$ intersect, we have that $U_a \geq L_c$, so $I_b \subseteq I_a \cup I_c$ and it must hold that $\pred{w}_b \in I_a$ or $\pred{w}_b \in I_c$, a contradiction to the instance being prediction mandatory free.
\end{proof}

Further, we observe that if the intervals of critical isolated vertices intersect intervals of vertices on such a path component, they must also intersect the interval of an endpoint of the component. 
Otherwise, the predicted weight $\w_v$ of the critical isolated vertex $v$ would be contained in the interval of at least one vertex on the path component, which contradicts the vertices on the path not being prediction mandatory.
The distinct partner $u$ of a critical isolated vertex $v$ that exists by~\Cref{lem:clique_partition_summary} is an endpoint of such a path component, as we show in~\Cref{predorient:sec:cliquepartition}.

The second phase of our algorithm iterates through all such connected components and, for each component, queries a minimum vertex cover (cf.~Lines~\ref{line:pathodd},~\ref{line:patheven1} and~\ref{line:patheven2}) and all resulting mandatory vertices (cf.~Line~\ref{line:mandatory2}).
If the path is of odd length, then the minimum vertex cover is unique. 
Otherwise, the algorithm selects the minimum vertex cover based on whether the interval of a critical isolated vertex intersects the interval of the first path endpoint.
This case would for example ensure that we pick the \enquote{right} vertex cover for the example instance of~\Cref{fig:ex:vcselection}.
Lemma~\ref{lem:VC-property} guarantees that the algorithm indeed queries a feasible query set.
The following lemma shows that this strategy indeed ensures $2$-robustness by exploiting~\Cref{lem:clique_partition_summary}.

\begin{lemma}
Algorithm~\ref{fig:sorting1cons2rob} is $2$-robust for sorting under explorable uncertainty with predictions.
\end{lemma}

\begin{proof}
Fix an optimum solution $\OPT$. 
Let $M_1$, $M_2$ and $\mathcal{S}$ denote the phase one queries of the algorithm as defined in the pseudocode.
Consider the clique partition $\mathcal{C}$ as computed in Line~\ref{line:cliquepartition}, 
then all $C \in \mathcal{C}$ with $|C| \ge 2$ satisfy $|C| \le 2 \cdot  |C \cap \OPT|$ and all $C \in \mathcal{C}$ with $C \subseteq M_1 \cup M_2$ satisfy $|C| \le |C \cap \OPT|$.
The latter holds as all members of $M_1 \cup M_2$ are mandatory by~\Cref{lema_mandatory_min}.
Queries to vertices that are covered by such cliques do not violate the $2$-robustness.
This leaves members of $\mathcal{S}$ that are  {critical} isolated  {vertices} 
in $\mathcal{C}$ and queries of the second phase. 
We partition such queries (and some non-queried vertices) into a collection~$\mathcal{W}$ such that, for each $W \in \mathcal{W}$, the algorithm queries at most $2 \cdot |W \cap \OPT|$ vertices in~$W$.
If we have such a partition, then it is clear that we spend at most $2 \cdot \opt$ queries and are $2$-robust.

By Lemma~\ref{lem:clique_partition_summary}, there is a distinct vertex $u \not\in M_1 \cup M_2 \cup \mathcal{S}$ for each  {critical} isolated vertex $v$ with $I_v \cap I_u \not= \emptyset$; as we noted above and will show in~\Cref{predorient:sec:cliquepartition}, $u$ is the endpoint of a path component of the current instance before line~\ref{line:vcbeginsort}. 
We create the partition $\mathcal{W}$ as follows: Iteratively consider all connected (path) components $P$ of the current instance before line~\ref{line:vcbeginsort}. 
Let $W$ be the union of $P$ and the  {critical} isolated vertices that are the distinct partner of at least one endpoint of $P$.
If $|W| \ge 2$, add $W$ to $\mathcal{W}$.
Then, $\mathcal{W}$ contains all  {critical} isolated vertices of~$\mathcal{C}$ and all vertices that are queried in the Lines~\ref{line:pathodd},~\ref{line:patheven1},~\ref{line:patheven2} and~\ref{line:mandatory2}.

We conclude the proof by arguing that each $W \in \mathcal{W}$ satisfies that the algorithm queries at most $2 \cdot |W \cap \OPT|$ vertices in~$W$.
By construction, $W$ contains a path component $P$ and up-to two  {critical} isolated vertices.
Furthermore, $W$ itself is a path in the initial interval graph
{(in addition to the edges of the path, there may be an additional edge between each critical
	isolated vertex of $\mathcal{C}$ in $W$ and the second or penultimate
	vertex of~$P$, but this does not affect the argument that follows)}. 
Consider an arbitrary $W \in \mathcal{W}$.
If $|W|$ is even, then $|W| \le 2 \cdot |W \cap \OPT|$ as all pairs of neighboring vertices in path $W$ are witness pairs.
Thus, assume that~$|W|$ is odd.
As each critical isolated vertex has a distinct partner by~\Cref{lem:clique_partition_summary} and this partner is an endpoint of a path component, $W$ contains at most one  {critical} isolated vertex per distinct endpoint of $P$ and, thus, we have $|P| \ge 2$.

We divide the analysis in two cases. First assume that $|P|=p$ is odd.
Then the algorithm queries $\{x_2, x_4, \ldots, x_{p-1}\}$ in Line~\ref{line:pathodd}.
As $P$ is a path, the precise weight of each queried vertex can be contained in the interval of at most one other vertex of $P$ and, therefore, force at most one query in Line~\ref{line:mandatory2}. 
This leaves at least one vertex in~$P\subseteq W$ that is never queried by the algorithm.
Since $|W \cap \OPT| \geq \lfloor |W|/2 \rfloor$ (as the subgraph induced by $W$ contains a path of the vertices in $W$), clearly the algorithm queries at most $2 \cdot |W \cap \OPT|$ vertices in~$W$.

Now assume that $|P|=p$ is even.
Then either~{$x_1$ or $x_p$} (but not both)
{is the distinct partner of a}
critical isolated member of $W$, otherwise~$|W|$ would be even.
If $I_{x_1}$ 
intersects the interval $I_v$ of the critical isolated vertex $v$, then the algorithm queries $\{x_1, x_3, \ldots, x_{p-1}\}$ in Line~\ref{line:patheven1}.
If~$w_{x_1}$ forces a query to~$x_2$ in Line~\ref{line:mandatory2} because $w_{x_1}\in I_{x_2}$, then $|\{x_1,x_2,v\}| \le 2 \cdot |\{x_1,x_2,v\} \cap \OPT|$ and the remaining vertices~$W' = W \setminus \{x_1,x_2,v\}$ form an even path, which implies $|W'| \le 2 \cdot |W' \cap \OPT|$ and, therefore $|W| \le 2 \cdot |W \cap \OPT|$.
If~$w_{x_1}$ forces no query to~$x_2$ in Line~\ref{line:mandatory2} because $w_{x_1} \not\in I_{x_2}$, then 
$|\{x_1,v\}| \le 2 \cdot |\OPT \cap \{x_1,v\}|$ and we analyze $W' = W \setminus \{x_1,v\}$ as in the subcase for odd $|P|$. Hence,  
the algorithm queries at most $2 \cdot |W \cap \OPT|$ intervals of $W$. 

If $I_{x_p}$ intersects 
the interval $I_v$ of critical isolated member $v$, then we can analyze $W$ analogously.
\end{proof}

\begin{lemma}
\label{thm_sorting_opt+kM}
Algorithm~\ref{fig:sorting1cons2rob} spends at most $\opt + k_M \le \opt + k_h$ queries.
\end{lemma}

\begin{proof}
We show that the algorithm spends at most $\opt +k_M$ queries. 
Theorem~\ref{Theo_hop_distance_mandatory_distance} then implies $\opt + k_M \le \opt + k_h$ .
Fix an optimum solution $\OPT$.
Every vertex queried in Lines~\ref{line:mandatoryproper} and~\ref{line:mandatory} is in $\OPT$ by~\Cref{lema_mandatory_min}.
Every vertex queried in Line~\ref{line:containpredicted} that is not in $\OPT$ is clearly in $\mathcal{I}_P \setminus \mathcal{I}_R$ and contributes one to $k_M$.

For each path~$P$ considered in Line~\ref{line:looppaths}, let $P'$ be the vertices queried in Lines~\ref{line:pathodd}--\ref{line:patheven2}.
It clearly holds that $|P'| \leq |P \cap \OPT|$.
Finally, every vertex queried in Line~\ref{line:mandatory2} is in $\mathcal{I}_R \setminus \mathcal{I}_P$, and therefore contributes to $k_M$, because we query all prediction mandatory vertices at the latest in Line~\ref{line:containpredicted}.
\end{proof}

\subsection{Computing the Clique Partition}
\label{predorient:sec:cliquepartition}

We continue by proving \Cref{lem:clique_partition_summary} and restate it here for the sake of readability:

\LemClique*

We describe how to compute a clique partition $\mathcal{C}$ of $\mathcal{I}_P \cup M$ that satisfies the lemma.
To that end, consider the set $\mathcal{I}_P \setminus M$. The elements of $M$ are allowed to be part of a clique of size one, so we ignore them for now.
Each $v \in \mathcal{I}_P \setminus M$ is prediction mandatory because it contains the predicted weight of some other vertex $u$ (cf.~\Cref{lema_mandatory_min} and recall that the input instance is an interval graph).
For each such $v$ we pick an arbitrary $\pi(v) \in V$ with $\w_{\pi(v)} \in I_v$ as the \emph{parent} of $v$.
Next, we build a forest of arborescences (rooted out-trees) that contains all vertices of $\mathcal{I}_P \setminus M$ (and some additional vertices) and define it by its arcs $\mathcal{E}$.
Afterwards, we use $\mathcal{E}$ to define the clique partition of the lemma.
We construct $\mathcal{E}$ by just iteratively considering all elements $v \in \mathcal{I}_P \setminus M$ in an arbitrary order and add $(\pi(v),v)$ to $\mathcal{E}$ if that does not create a cycle. Each so constructed arborescence contains at most one vertex that is not in $\mathcal{I}_P \cup M$, and this vertex has to be the root of the arborescence.

\begin{algorithm}[tb]
\KwIn{Forest of arborescences~$\mathcal{E}$, set of prediction mandatory vertices $\mathcal{I}_P$, set of known mandatory vertices $M$.}
$C_v \gets \emptyset$ for all $v \in V$;
$\mathcal{S} \gets \mathcal{I}_P \cup M$\;
\While{$\mathcal{S} \neq \emptyset$\label{line:cliquepartitionbegin}}{
	\KwLet $v$ be a deepest vertex in the forest of arborescences $(\mathcal{I}, \mathcal{E})$ among those in~$\mathcal{S}$\;
	\If{$(\pi(v), v) \in \mathcal{E}$}{
		$C_{\pi(v)} \leftarrow \{v' \in \mathcal{S} : (\pi(v), v') \in \mathcal{E}\}$\;
		\lIf{$\pi(v) \in \mathcal{S}$}{$C_{\pi(v)} \leftarrow C_{\pi(v)} \cup \{\pi(v)\}$}
		$\mathcal{S} \leftarrow \mathcal{S} \setminus C_{\pi(v)}$\label{line:cliquepartitionend}\;
	}
	\lElse{$C_v \leftarrow \{v\}$; \quad $\mathcal{S} \leftarrow \mathcal{S} \setminus  {C_v}$\label{line:rootisolated}}
}
\Return $\mathcal{C} = \{C_v \mid v \in V \land C_v \not= \emptyset\}$ \;
\caption{Algorithm to create a clique partition from a forest of arborescences.}
\label{alg:clique-partition}
\end{algorithm}

Based on $\mathcal{E}$, we create a first clique partition $\mathcal{C}$ using Algorithm~\ref{alg:clique-partition}.
Since all vertices in set $C_v$, as created by the algorithm, contain the predicted weight $\w_v$, the created sets are clearly cliques.
Each of the partial clique partitions (created for a single arborescence in the forest of arborescences) may contain at most a single clique of size one. 
To satisfy the lemma, each such clique $\{v\}$ must either satisfy $v \in M$ or needs a distinct $u \not\in \mathcal{I}_P \cup M$ with $I_v \cap I_u \not= \emptyset$.

If the root of the arborescence is in $M$, then the only clique $\{v\}$ of size one created for the arborescence contains the root of the arborescence. Since the root is in $M$, the cliques for the arborescence satisfy the lemma.
If the root of the arborescence is not in $\mathcal{I}_P \cup M$, then the vertex $v$ in the clique of size one might not be the root of th arborescence but a direct child of the root.
In that case, the parent $\pi(v)$ of $v$ (the vertex in the clique of size~$1$) is not in $\mathcal{I}_P \cup M$ and we use $\pi(v)$ as the distinct partner $u \not\in \mathcal{I}_P \cup M$ with $I_v \cap I_u \not= \emptyset$ of $v$.
We remark that there must be such a $\pi(v)$ that is the endpoint of a path component in the subgraph induced by $V\setminus (\mathcal{I}_P \cup M)$ as otherwise there must be a vertex $u$ on the path with $\w_v \in I_u$; a contradiction to the vertices on the path not being part of $\mathcal{I}_P$. We pick this endpoint of a path component as the partner of $v$.
In case we run into this situation for multiple cliques $\{v\}$ of different arborescences but with the same $\pi(v)$, we can just merge all those cliques into a single one. This results in a clique as the intervals of all vertices in those smaller cliques contain the predicted weight $\w_{\pi(v)}$.

The problematic case is when the root of the arborescence is part of $\mathcal{I}_P \setminus M$.
Note that this can only be the case if the parent $\pi(r)$ of the root $r$ is also part of the arborescence, as otherwise the corresponding edge $(\pi(r),r)$ would have been added and $r$ would not be the root. So the parent $\pi(r)$ of the root $r$ is also part of the arborescence but the edge $(\pi(r),r)$ was not added because it would have created a cycle.
We handle this situation by showing the following auxiliary lemma.
It states that in this case we can revise the clique
partition of that arborescence in such a way that all cliques in that clique partition for the have size at least~$2$.
The auxiliary lemma then concludes the proof of \Cref{lem:clique_partition_summary}.

\begin{restatable}{lemma}{cliquerepartition}
\label{lem:cliquerepartition}%
Consider an out-tree (arborescence) $T$ on a set of prediction mandatory
vertices, where an edge $(u,v)$ represents that $\pred{w}_u\in I_v$.
Let the root be $r$.
Let vertex $m$ with $\pred{w}_m\in I_r$ be a descendant
of the root somewhere in $T$. Then the vertices in $T$ can be partitioned into
cliques (sets of pairwise overlapping intervals) in such a way
that all cliques have size at least~$2$.
\end{restatable}


\begin{proof}
We refer to the clique partition method of Lines~\ref{line:cliquepartitionbegin}--\ref{line:rootisolated} in Algorithm~\ref{alg:clique-partition}
as algorithm~CP.
This method will partition the nodes of an arborescence into cliques, each consisting either of a subset of the children of a node, or of a subset of the children of a node plus the parent of those children.
In the case considered in this lemma,
all cliques will have size at least~$2$, except that the clique containing the root of the tree may have size~$1$.

We first modify $T$ as follows: If there is a node $v$ in $T$ that is not a child of the root $r$ but contains $\pred{w}_r$, then we make $r$ the parent of $v$ (i.e., we remove the subtree rooted at $v$ and re-attach it below the root).
After this transformation, all vertices whose intervals contain $\pred{w}_r$ are children of $r$.

Apply CP to each subtree of $T$ rooted at a child of $r$.
For each of the resulting partitions, we call the clique containing
the root of the subtree the \emph{root clique} of that subtree.
There are several possible outcomes that can be handled directly:
\begin{itemize}
	\item At least one of the clique partitions has a root clique
	of size~$1$. In that case we combine all these root cliques
	of size~$1$ with $r$ to form a clique of size at least~$2$, and we are done:
	This new clique together with all remaining cliques from the clique
	partitions of the subtrees forms the desired clique partition.
	
	\item All of the clique partitions have root cliques of size at least~$2$,
	and at least one of them has a root clique of size at least $3$.
	Let $s$ be the root node of a subtree whose root clique has
	size at least~$3$. We remove~$s$ from its clique and form
	a new clique from $s$ and $r$, and we are done.
	
	\item All of the clique partitions have root cliques of size exactly~$2$,
	and at least one of the children $v$ of $r$ has $\pred{w}_v\in I_r$.
	Then we add $r$ to the root clique that contains~$v$. We can do
	this because all intervals in that root clique contain $\pred{w}_v$.
\end{itemize}

Now assume that none of these cases applies, so we have the following
situation:
All of the clique partitions have root cliques of size exactly~$2$,
and every child $v$ of $r$ has its predicted weight outside $I_r$,
i.e., $\pred{w}_v\notin I_r$. In particular, $m$, the vertex whose interval makes $r$ prediction mandatory ($\pred{w}_m \in I_r$), cannot be a child of $r$.

\begin{figure}[t]
	\centerline{\scalebox{0.6}{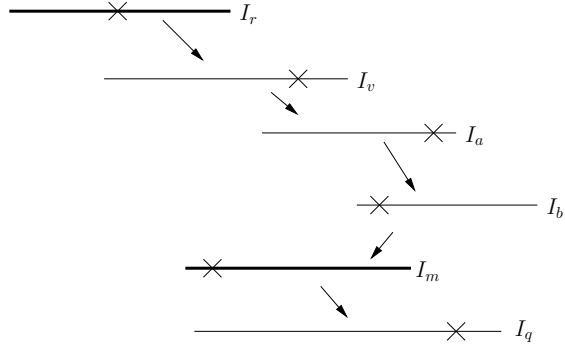}}
	\caption{Illustration of path from $I_r$ to $I_m$'s child $I_q$ in $T$}
	\label{fig:tree}
\end{figure}

Let $T'$ be the subtree of $T$ that is rooted at a child
of $r$ and that contains $m$. Let the root of $T'$
be $v$.

Observe that $I_v$ is the only interval in $T'$ that
contains $\pred{w}_r$, because all vertices with intervals containing $\pred{w}_r$ are
children of $r$ in~$T$. Assume w.l.o.g.\ that $\pred{w}_v$ lies
to the right of $I_r$. Then all intervals of vertices in~$T'$,
except for $I_v$, lie to the right of~$\pred{w}_r$.
See Figure~\ref{fig:tree} for an illustration of
a possible configuration of the intervals on the path from
$r$ to $m$ (and a child $q$ of~$m$) in~$T$.

Now re-attach the subtree $T_m$ rooted at $m$
as a child of $r$ (ignoring the fact that $\pred{w}_r$ is
not inside $I_m$), and let $T_v=T'\setminus T_m$ denote
the result of removing $T_m$ from $T'$. Re-apply CP
to the two separate subtrees $T_m$ and $T_v$.
The possible outcomes are:
\begin{itemize}
	\item The root clique of at least one of the two subtrees has size~$1$.
	We can form a clique by combining $r$ with those (one or two)
	root cliques of size~$1$.
	As both $I_v$ and $I_m$ intersect $I_r$ from the right, the
	resulting set is indeed a clique. Together with all other
	cliques from the clique partitions of $T_m$ and $T_v$,
	and those of the other subtrees of $r$ in $T$,
	we obtain the desired clique partition.
	\item The root cliques of both subtrees have size at
	least~$2$. We add $r$ to the root clique of $T_m$.
	That root clique contains only vertices with intervals containing
	$\pred{w}_m$, and $I_r$ also contains $\pred{w}_m$, so
	we do indeed get a clique if we add $r$ to that
	root clique. This new clique, together with all other
	cliques from the clique partitions of $T_m$ and $T_v$,
	and those of the other subtrees of $r$ in $T$,
	forms the desired clique partition.
\end{itemize}

This concludes the proof of the lemma.
\end{proof}

\subsection{The $k_{\#}$-dependent guarantee}

We continue by proving that Algorithm~\ref{fig:sorting1cons2rob} executes at most $\opt + k_{\#}$ queries. This then concludes the proof of~\Cref{thm:sorting:singleset}.

Recall that, in order to compute the clique partition of~\Cref{lem:cliquerepartition}, we, for each $v \in \mathcal{I}_P \setminus M$ fix an arbitrary $\pi(v) \in V$ with $\w_{\pi(v)} \in I_v$ as the \emph{parent} of $v$.
These parents will be used in the following proofs.

To prove the $k_{\#}$-dependent guarantee, we need the following auxiliary lemma.

\begin{lemma}
\label{lem:endpoints}%
Fix the state of set $\mathcal{S}$ as in Line~\ref{line:fixS} of Algorithm~\ref{fig:sorting1cons2rob}.
For each vertex $u$, let $S_u = \{ v \in \mathcal{S} : \pi(v) = u \}$.
For any path~$P$ considered in Line~\ref{line:looppaths} and any vertex $u \in P$ that is not an endpoint of~$P$, it holds that $S_u = \emptyset$.
\end{lemma}

\begin{proof}
Suppose for the sake of contradiction that some $u \in P$ that is not an endpoint of~$P$ has $S_u \neq \emptyset$.
Let~$a$ and~$b$ be its neighbors in~$P$, and let $v \in S_u$.
We have that $\pred{w}_u \notin I_a \cup I_b$, otherwise either~$a$ or~$b$ would be prediction mandatory and, therefore, have been queried in Line~\ref{line:containpredicted}.
Thus~$(I_v \cap I_u) \setminus (I_a \cup I_b) \neq \emptyset$, because $\pred{w}_u \in I_v$ but $\w_u \not\in I_a \cup I_b$.
It is not the case that $I_v \subseteq I_u$ or $I_u \subseteq I_v$: If $I_v \subseteq I_u$, then~$u$ would have been queried in Line~\ref{line:mandatoryproper} before~$P$ is considered; if $I_u \subseteq I_v$, then~$v$ would have been queried in Line~\ref{line:mandatoryproper} and $v \notin \mathcal{S}$.
Therefore it must be that $I_v \subseteq I_a \cup I_u \cup I_b$, otherwise $I_a \subseteq I_v$ or $I_b \subseteq I_v$ (again a contradiction for $v \in \mathcal{S}$) since $(I_v \cap I_u) \setminus (I_a \cup I_b) \neq \emptyset$ and $a,u,b$ forms a simple path in the interval graph.
However, if $I_v \subseteq I_a \cup I_u \cup I_b$, then~$v$ would have forced a query to~$a$,~$b$ or~$u$ in Line~\ref{line:mandatory} as one of the corresponding intervals must contain $w_v$ and, thus, becomes mandatory by Corollary~\ref{cor_min_left_mandatory}.
This is a contradiction to $a$,$b$ and $u$ being part of path $P$.
\end{proof}

\begin{theorem}
Algorithm~\ref{fig:sorting1cons2rob} performs at most $\opt + k_{\#}$ queries.
\end{theorem}

\begin{proof}
Fix an optimum solution~$\OPT$.
We partition the vertices in~$V$ into sets with the following properties.
One of the sets~$\tilde{S}$ contains vertices that are not queried by the algorithm.
We have a collection~$\mathcal{S}'$ of vertex sets in which each set has at most one vertex not in $\OPT$, i.e., $|S' \setminus \OPT| \le 1$ for all $S' \in \mathcal{S}'$.
Also, if it has one vertex not in~$\OPT$, then we assign a distinct prediction error to that set, in such a way that each error is assigned to at most one set.
The vertex corresponding to the prediction error does not need to be in the same set.
Let~$V'$ be the set of vertices with a prediction error ($w_v \not= \w_v$ for all $v \in V'$) assigned to some set in~$\mathcal{S}'$.
Finally, we have a collection~$\mathcal{W}$ of vertex sets such that for every $W \in \mathcal{W}$ it holds that $|\ALG \cap W| \leq |W \cap \OPT| + k_{\#}(W \setminus V')$, where $k_{\#}(X)$ is the number of vertices in $X$ with incorrect predictions.
If we have such a partition, then it is clear that we spend at most $\opt + k_{\#}$ queries.

We begin by adding a set that contains all vertices queried in Lines~\ref{line:mandatoryproper} and~\ref{line:mandatory} ($M_1 \cup M_2$ in the Pseudocode) to $\mathcal{S}'$; all such vertices are clearly in $\OPT$, and we do not need to assign a prediction error.

Fix the state of $\mathcal{S}$ as in Line~\ref{line:fixS}.
To deal with the vertices queried in Line~\ref{line:containpredicted}, we add to~$\mathcal{S}'$ the set $S_u = \{v \in \mathcal{S} : \pi(v) = u\}$ for all~$u \in V$.
Note that each such set is a clique, because the corresponding intervals of all vertices in $S_u$ contain~$\pred{w}_u$ and we are considering an interval graph.
Therefore, using Lemma~\ref{lem:witness}, at most one vertex in~$S_u$ is not in $\OPT$, and if that occurs, then $\pred{w}_u \neq w_u$, and we assign this prediction error to~$S_u$.

Let~$P = x_1 x_2 \cdots x_p$ with $p \geq 2$ be a path considered in Line~\ref{line:looppaths}, and let~$P'$ be the set of vertices in~$P$ that are queried in the following execution of  Lines~\ref{line:pathodd},~\ref{line:patheven1} or~\ref{line:patheven2}.
It clearly holds that $|P'| = \lfloor |P| / 2 \rfloor \leq |P \cap \OPT|$ as $P$ is a path and each edge defines a witness set~\Cref{lem:witness}.
It also holds that at most $k_{\#}(P')$ intervals in~$P$ are queried in Line~\ref{line:mandatory2}: Each vertex $u \in P'$ can force a query to at most one vertex~$v$ in Line~\ref{line:mandatory2} as weight $w_u$ can only be contained in the interval of at most one neighbor $v$ of $u$ in the path, and in that case the predicted weight of~$u$ is incorrect because $w_u \in I_v$ but $\pred{w}_u \notin I_v$, or $v$ would have been queried in Line~\ref{line:containpredicted}.
We will create a set $W \in \mathcal{W}$ and possibly modify~$\mathcal{S}'$, in such a way that $P \subseteq W$ and $P' \cap V' = \emptyset$, so it is enough to show that
\begin{equation}
	|\ALG \cap W| \leq |W \cap \OPT| + k_{\#}(P').\label{eq:sortingpath}
\end{equation}
We initially take $W$ as the vertices in~$P$.
By Lemma~\ref{lem:endpoints}, it holds that $S_u = \emptyset$ for any $u \in P$ with $u \not\in \{x_1, x_p\}$.
If $S_{x_1} \subseteq \OPT$, then we did not assign a prediction error to~$S_{x_1}$.
Otherwise, let~$v$ be the only vertex in $S_{x_1} \setminus \OPT$.
The predicted weight of $I_{x_1}$ is incorrect because $\pred{w}_{x_1} \in I_v$, and it must hold that $x_1 \in \OPT$, or $\OPT$ would not be able to decide the order between~$x_1$ and~$v$ (as $I_v \cap I_{x_1} \not= \emptyset$ which implies that $\{v,x_1\}$ is an edge and, therefore, a witness pair).
If $x_1 \notin P'$, then we will not use its error in the bound of $|\ALG \cap W|$ to prove Equation~(\ref{eq:sortingpath}).
Otherwise, we add~$v$ to~$W$ and remove it from~$S_{x_1}$, and now we do not need to assign a prediction error to~$S_{x_1}$ anymore as $v$ was the sole member of $S_{x_1}$ not in $\OPT$.
We do a similar procedure for~$x_p$, and since at most one of $x_1, x_p$ is in~$P'$ by definition of Lines~\ref{line:pathodd},~\ref{line:patheven1} and~\ref{line:patheven2}, we only have two cases to analyze: (1)~$W = P$, or (2)~$W = P \cup \{v\}$ with $\pi(v) \in \{x_1, x_p\}$.
\begin{enumerate}
	\item $W = P$. Clearly $|\ALG \cap W| \leq |P'| + k_{\#}(P') \leq |W \cap \OPT| + k_{\#}(P')$ as we already argued that each member of $P'$ (the vertices queried in Lines~\ref{line:pathodd},~\ref{line:patheven1} or~\ref{line:patheven2})
	can lead to at most one additional query in Line~\ref{line:mandatory2} and only if it has an incorrect predicted weight. Since in this case  $|P|$ must either be odd or $x_1 \in P'$ but $S_{x_1}\subseteq \OPT$ or $x_p \in P'$ but $S_{x_p} \subseteq \OPT$, we have $P' \cap V' = \emptyset$.
	
	\item $W = P \cup \{v\}$, with $\pi(v) \in \{x_1, x_p\}$.
	Suppose w.l.o.g.\ that $\pi(v) = x_1$.
	Remember that $x_1 \in P'$, that $x_1 \in \OPT$, that $v \not\in \OPT$ and that the predicted weight of $x_1$ is incorrect.
	Since $x_1 \in P'$, it holds that $|P|$ is even and $x_2 \notin P'$.
	We have two cases.
	\begin{enumerate}
		\item $x_2$ is not queried in Line~\ref{line:mandatory2}.
		Then $x_1$ does not force a query in Line~\ref{line:mandatory2}, so
		\begin{eqnarray*}
			|\ALG \cap W| & \leq & |P' \cup \{v\}| + k_{\#}(P' \setminus \{x_1\}) \\
			& = & |P'| + 1 + k_{\#}(P' \setminus \{x_1\}) \\
			& \leq & |P \cap \OPT| + k_{\#}(P') \\
			& \leq & |W \cap \OPT| + k_{\#}(P').
		\end{eqnarray*}
		
		\item $x_2$ is queried in Line~\ref{line:mandatory2}.
		Then $x_1, x_2 \in \OPT$, and $|\OPT \cap (P \setminus \{x_1, x_2\})| \geq |P' \setminus \{x_1\}|$ because $|P|$ is even.
		Therefore,
		\begin{eqnarray*}
			|\ALG \cap W| & \leq & |P' \cup \{x_2, v\}| + k_{\#}(P' \setminus \{x_1\}) \\
			& \leq & |P \cap \OPT| + 1 + k_{\#}(P' \setminus \{x_1\}) \\
			& \leq & |W \cap \OPT| + k_{\#}(P').   
		\end{eqnarray*}
	\end{enumerate}
\end{enumerate}

Finally, we add the remaining intervals that are not queried by the algorithm to $\tilde{S}$.
\end{proof}

\section{Learnability of predictions}
\label{sec:learnability}

In this section, we argue about the learnability of our predictions with regard to the different error measures for a given instance of hypergraph orientation $H=(V,E)$ with the set of uncertainty intervals $\mathcal{I} = \{I_v \mid v \in V\}$.

We assume that the realization $w$ of precise weights for $\mathcal{I}$ is i.i.d.~drawn from an unknown distribution $\ud$, and 
that we can i.i.d.~sample realizations from $\ud$ 
to obtain a training set.
Let $\hs$ denote the set of all possible
prediction vectors $\pred{w}$, 
with $\pred{w_v} \in I_v$ for each $I_v \in \mathcal{I}$.
Let $k_h(w,\pred{w})$ denote the hop distance of the prediction $\pred{w}$ for the realization with the precise weights $w$.
Since $w$ is drawn from $\ud$, the value $k_h(w,\pred{w})$ is a random variable. 
Analogously, we consider $k_M(w,\pred{w})$ with regard to the mandatory query distance.
Our goal is to learn predictions $\pred{w}$ that (approximately) minimize the expected error $\EX_{w \sim \ud}[k_h(w,\pred{w})]$ respectively $\EX_{w \sim \ud}[k_M(w,\pred{w})]$.
In the following, we argue separately about the learnability with respect to $k_h$ and $k_M$.

Note that, under the PAC-learning assumption that we can sample the precise weights, it is also possible to apply the algorithm for the stochastic problem variant given in~\cite{BampisDEdLMS21}. However, in contrast to the results of this paper, the guarantee of the algorithm given in~\cite{BampisDEdLMS21} holds in expectation instead of in the worst-case and admits no robustness. Furthermore, our algorithms work as a blackbox independent of how the predictions are obtained.

\subsection{Learning with respect to the hop distance}

The following theorem states the learnability w.r.t.~$k_h$ and directly follows from the (problem independent) learnability proof in~\cite{EdLMS22}.

\begin{restatable}{theorem}{TheoLearningHop}
	\label{theo_learnability_hop}
	For any~$\eps, \delta \in (0,1)$, there exists a learning algorithm that, using a training set of size~$m$, returns predictions $\pred{w} \in \hs$, such that
	$\EX_{w \sim \ud}[k_h(w,\pred{w})] \le \EX_{w \sim \ud}[k_h(w,\pred{w}^*)] + \eps$ holds with probability at least~$(1-\delta)$, where~$\pred{w}^* = \arg\min_{\pred{w}' \in \hs} \EX_{w \sim \ud}[k_h(w,\pred{w}')]$.
	The sample complexity is $m \in \mathcal{O}\left(\frac{(\log(n) - \log(\delta/n))\cdot (2n)^2}{(\eps/n)^2}\right)$ and the running time is polynomial in $m$ and $n$.

\end{restatable}

\subsection{Learning with respect to the mandatory query distance}

We argue about the learnability w.r.t.~$k_M$.
Since each $I_v$ is an open interval, there are infinitely many predictions $\pred{w}$, and, thus, the set $\hs$ is also infinite.
In order to reduce the size of $\hs$, we discretize each $I_v$ by fixing a finite number of potentially predicted weights $\pred{w}_v$ of $I_v$ using the same technique as in~\cite{EdLMS22}.

We define the set $\hs_v$  of 
predicted weights for $v$ as follows.
Let $\{B_1,\ldots,B_l\}$ be the set of lower and upper limits of intervals in $\mathcal{I} \setminus \{I_v\}$ that are contained in $I_v$.
Assume that $B_1,\ldots,B_l$ are indexed by increasing value. 
Let $B_0 = L_v$ and $B_{l+1} = U_v$ and, for each $j \in \{0,\ldots,l\}$, let $h_j$ be an arbitrary value of $(B_j,B_{j+1})$.
We define $\hs_v = \{B_1,\ldots,B_l,h_0,\ldots,h_l\}$.
Since two weights $\pred{w}_v,\pred{w}_v' \in (B_j,B_{j+1})$ always lead to the same hop distance $\jo(v)$ for vertex $v$, there will always be an element of $\hs_v$ that minimizes the expected hop distance for $v$.
Each $\hs_v$ contains at most $\mathcal{O}(n)$ values.
Then, the discretized $\hs$ is the Cartesian product of all discretized $\hs_v$ with $v \in V$ and, thus, the distcretzation reduces the size of $\hs$ to at most $\mathcal{O}(n^n)$.

The following lemma shows that we do not lose any precision by using the discretized $\hs$.
Here, for any vector $\pred{w}$ of predicted weights, $\mathcal{I}_{\pred{w}}$ denotes the set of prediction mandatory vertices.
In particular, $\hs$ is now finite.

\begin{lemma}
	\label{lem_discretization}
	For a given instance of the hypergraph orientation problem with intervals $\mathcal{I}$, let $\pred{w}$ be a vector of predicted weights that is not contained in the discretized $\hs$.
	Then, there is a $\pred{w}' \in \hs$ such that $\mathcal{I}_{\pred{w}} = \mathcal{I}_{\pred{w}'}$.
\end{lemma}

The lemma implies that, for each $\pred{w}$, there is an $\pred{w}' \in \hs$ that has the same error w.r.t.~$k_M$ as $\pred{w}$.
Thus, there always exists an element $\pred{w}$ of the discretized $\hs$ such that $\pred{w}$ minimizes the expected error over all possible vectors of predicted weights. 

\begin{proof}[Proof of Lemma~\ref{lem_discretization}]	
	Given the vector of predicted weights $\pred{w}$, we construct a vector $\pred{w}' \in \hs$ that satisfies the lemma.

	For each $\pred{w}_v$, we construct $\pred{w}'_v$ as follows:
	If $\pred{w}_v = B_j$ for some $j \in \{1,\ldots,l\}$, then we set $\pred{w}_v' = B_j$, where $B_j$ is defined as in the definition of the discretized $\hs_v$.
	Otherwise it must hold $\pred{w}_v \in (B_j,B_{j+1})$ for some $j \in \{1,\ldots,l\}$, and we set $\pred{w}'_v = h_j$.
	Then, for each $u \in V\setminus \{v\}$, it holds $\pred{w}_v \in I_u$ if and only if $\pred{w}_v' \in I_u$.
	
	We show that each $v \in \mathcal{I}_{\pred{w}}$ is also contained in $\mathcal{I}_{\pred{w}'}$.
	Since $v$ is mandatory assuming precise weights $\pred{w}$, Lemma~\ref{lema_mandatory_min} implies that there is a hyperedge $S$ such that either (i) $\pred{w}_v$ is the minimum weight of $S$ and $\pred{w}_u \in I_v$ for some $u \in S\setminus\{v\}$ or (ii) $\pred{w}_v$ is not the minimum weight of $S$ but contains the weight $\pred{w}_u$ of the vertex $u$ with minimum weight $\pred{w}_u$ in $S$.
	
	Assume $v$ satisfies case (i) for the predicted weights $\pred{w}$.
	By construction of $\pred{w}'$ it then also holds $\pred{w}'_u \in I_v$.
	Thus, if $v$ has minimum weight in $S$ for the weights $\pred{w}'$, then $v \in \mathcal{I}_{\pred{w}'}$ by Lemma~\ref{lema_mandatory_min}.
	Otherwise, some $v' \in S\setminus\{v\}$ must have the minimum weight $\pred{w}'_{v'}$ in $S$ for the weights $\pred{w}'$.
	Since $v$ has minimum weight for the weights $\pred{w}$, it must hold $\pred{w}'_{v'} < \pred{w}'_v$ but $\pred{w}_{v'} \ge \pred{w}_v$.
	By construction of $\pred{w}'$, this can only be the case if $\pred{w}'_{v'},\pred{w}_{v'} \in I_v$.
	This implies that $v$ satisfies case (ii) for the weights $\pred{w}'$ and, therefore $v \in \mathcal{I}_{\pred{w}'}$ by Lemma~\ref{lema_mandatory_min}.
	
	Assume $v$ satisfies case (ii) for the predicted weights $\pred{w}$.
	By construction of $\pred{w}'$ it then also holds $\pred{w}'_u \in I_v$.
	Thus, if $u$ has minimum weight in $S$ for the weights $\pred{w}'$, then $v \in \mathcal{I}_{\pred{w}'}$.
	Otherwise, some $u' \in S\setminus\{u\}$ must have minimum weight in $S$ for the weights $\pred{w}'$.
	If $u' = v$, then $v$ satisfies case (i) for the weights $\pred{w}'$ and, therefore, $v \in \mathcal{I}_{\pred{w}'}$ by Lemma~\ref{lema_mandatory_min}.
	If $u'\not= u$, then it must hold $\pred{w}'_{u'} < \pred{w}'_u$ but $\pred{w}_{u'} \ge \pred{w}_u$ as $u$ has minimum weight in $S$ for weights $\pred{w}$ but $u'$ has minimum weight in $S$ for weights $\pred{w}'$.
	By construction and since $\pred{w}_u \in I_v$, this can only be the case if $\pred{w}'_{u'},\pred{w}_{u'} \in I_v$.
	This implies that $v$ satisfies case (ii) for the weights $\pred{w}'$ and therefore $v \in \mathcal{I}_{\pred{w}'}$ by Lemma~\ref{lema_mandatory_min}.
	
	Symmetrically, we can show that each $v \in \mathcal{I}_{\pred{w}'}$ is also contained in $\mathcal{I}_{\pred{w}}$, which implies $\mathcal{I}_{\pred{w}} = \mathcal{I}_{\pred{w}'}$.
\end{proof}

Recall that $k_M$ is defined as $k_M = |\mathcal{I}_P \sym \mathcal{I}_R|$, where $\mathcal{I}_P$ is the set of predictions mandatory vertices and $\mathcal{I}_R$ is the set of mandatory vertices.
In contrast to learning predictions $\pred{w}$ w.r.t.~$k_h$, a vertex	 $v \in \mathcal{I}$ being part of $\mathcal{I}_P \sym \mathcal{I}_R$ depends not only on $\pred{w}_v$ and $w_v$, but on the predicted and precise weights of vertices $V \setminus \{v\}$. 
Thus, the events of $v$ and $u$ with $v \not= u$ being part of $\mathcal{I}_P \sym \mathcal{I}_R$ are not necessarily independent.
Therefore, we cannot separately learn the predicted weights $\pred{w}_v$ for each $v \in V$, which is a major difference to the proof for $k_h$ in~\cite{EdLMS22}.
The approach in~\cite{EdLMS22} learns the predicted weights separately in order to achieve a polynomial running time.

Since the discretized $\hs$ is still finite and $k_M$ is bounded by $[0,n]$, we still can apply \emph{empirical risk minimization (ERM)} to achieve guarantees similar to the ones of Theorem~\ref{theo_learnability_hop}.
However, because we cannot learn the $\pred{w}_v$ separately, we would have to find the element of $\hs$ that minimizes the empirical error. 
ERM first i.i.d.~samples a trainingset~$S=\{w^1,\ldots,w^m\}$ of~$m$ precise weights vectors from~$\ud$.
Then, it returns the predicted weights~$\pred{w} \in \hs$ that minimizes the \emph{empirical error}~$k_S(\pred{w}) = \frac{1}{m} \sum_{j=1}^{m} k_M(w^j,\pred{w})$.
As $\hs$ is of exponential size, a straightforward implementation of ERM requires exponential running time.
We use this straightforward implementation to prove the following theorem.

\begin{theorem}
	\label{theo_learnability_man_exp}
	For any~$\eps, \delta \in (0,1)$, there exists a learning algorithm that, using a training set of size~$m$, returns predictions $\pred{w} \in \hs$, such that
	$\EX_{w \sim \ud}[k_M(w,\pred{w})] \le \EX_{w \sim \ud}[k_M(w,\pred{w}^*)] + \eps$ holds with probability at least~$(1-\delta)$, where~$\pred{w}^* = \arg\min_{\pred{w}' \in \hs} \EX_{w \sim \ud}[k_M(w,\pred{w}')]$.
	The sample complexity is $m \in \mathcal{O}\left(\frac{(n \cdot \log(n) - \log(\delta))\cdot n^2}{\eps^2}\right)$ and the running time is exponential in $n$.
\end{theorem}

\begin{proof}
	Since $|\hs| \in \mathcal{O}(n^n)$ and $k_M$ is bounded by $[0,n]$, ERM achieves the guarantee of the theorem with a sample complexity of $m \in \mathcal{O}\left(\frac{(n \cdot \log(n) - \log(\delta))\cdot (n)^2}{(\eps)^2}\right)$; for details we refer to~\cite{Shalev2014,Vapnik1992}.
	The prediction $\overline{w} \in \hs$ that minimizes the empirical error $k_S(\pred{w}) = \frac{1}{m} \sum_{j=1}^{m} k_M(w^j,\pred{w})$, where $S=\{w^1,\ldots,w^m\}$ is the training set, can be computed by iterating through all elements of $S \times H$.
	Thus, the running time of ERM is polynomial in $m$ but exponential in $n$.
\end{proof}

To circumvent the exponential running time, we present an alternative approach.
In contrast to $k_h$, for a fixed realization, the value $k_M$ only depends on $\mathcal{I}_P$.
Instead of showing the learnability of the predicted weights, we prove that the set $\mathcal{I}_P$ that leads to the smallest expected error can be (approximately) learned.
To be more specific, let $\mathcal{P}$ be the power set of $V$, let $\mathcal{I}_w$ denote the set of mandatory vertices for the realization with precise weights $w$, and let $k_M(\mathcal{I}_w, P)$ with $P \in \mathcal{P}$ denote the mandatory query distance under the assumption that $\mathcal{I}_P = P$ and $\mathcal{I}_R = \mathcal{I}_w$.
Since $w$ is drawn from $\ud$, the value $\EX_{w \sim \ud}[k_M(\mathcal{I}_w,P)]$ is a random variable.
As main result of this section, we show the following theorem.

\begin{restatable}{theorem}{TheoLearningMan}
	\label{theo_learnability_mandatory}
	For any~$\eps, \delta \in (0,1)$, there exists a learning algorithm that, using a training set of size~$m  \in \mathcal{O}\left(\frac{(n - \log(\delta))\cdot n^2}{\eps^2}\right),$ returns a predicted set of mandatory intervals $P \in \mathcal{P}$ in time polynomial in~$n$ and~$m$, such that
	$\EX_{w \sim \ud}[k_M(\mathcal{I}_w,P)] \le \EX_{w \sim \ud}[k_M(\mathcal{I}_w,P^*)] + \eps$ holds with probability at least~$(1-\delta)$, where~$P^* = \arg\min_{P' \in \mathcal{P}} \EX_{w \sim \ud}[k_M(\mathcal{I}_w,P')]$.
\end{restatable}

Note that this theorem only allows us to learn a set of prediction mandatory vertices $P$ that (approximately) minimizes the expected mandatory query distance. 
It does not, however, allow us to learn the predicted weights $\pred{w}$ that lead to the set of prediction mandatory vertices $P$.
In particular, it can be the case, that no such predicted weigths exist. 
Thus, there may not be a realization with precise weights $w$ and $k_M(\mathcal{I}_w,P)=0$.
On the other hand, learning $P$ already allows us to execute Algorithm~\ref{ALG_min_alpha} for the hypergraph orientation problem.
Applying Algorithm~\ref{fig:sorting1cons2rob} for the sorting problem would require knowing the corresponding predicted weights $\pred{w}$.

\begin{proof}[Proof of Theorem~\ref{theo_learnability_mandatory}]
	We again show that the basic \emph{empirical risk minimization (ERM)} algorithm already satisfies the lemma. 
	ERM first i.i.d.~samples a trainingset~$S=\{w^1,\ldots,w^m\}$ of~$m$ precise weight vectors from~$\ud$.
	Then, it returns the~$P \in \mathcal{P}$ that minimizes the \emph{empirical error}~$k_S(P) = \frac{1}{m} \sum_{j=1}^{m} k_M(\mathcal{I}_{w^j},P)$.
	Since $\mathcal{P}$ is of exponential size, i.e., $|\mathcal{P}| \in \mathcal{O}(2^n)$, we cannot afford to naively iterate through $\mathcal{P}$ in the second stage of ERM, but have to be more careful.
	
	By definition, $\mathcal{P}$ contains $\mathcal{O}(2^n)$ elements and, thus, is finite.
	Since~$\mathcal{P}$ is finite, and the error function $k_M$ is bounded by the interval $[0,n]$, it satisfies the \emph{uniform convergence property} (cf.~\cite{Shalev2014}).
	This implies that, for~$$m = \left\lceil \frac{2\log(2|\mathcal{P}|/\delta)n^2}{\eps^2} \right\rceil \in \mathcal{O}\left(\frac{(n - \log(\delta))\cdot n^2}{\eps^2}\right),$$ it holds~$\EX_{w \sim \ud}[k_M(\mathcal{I}_{w},P)] \le \EX_{w \sim \ud}[k_M(\mathcal{I}_w,P^*)] + \eps$ with probability at least~$(1-\delta)$, where~$P$ is the set $P\in\mathcal{P}$ learned by ERM (cf.~\cite{Shalev2014,vapnik1999}).

	It remains to show that we can compute the set $P \in \mathcal{P}$ that minimizes the empirical error~$k_S(P) = \frac{1}{m} \sum_{j=1}^{m} k_M(\mathcal{I}_{w^j},P)$ in time polynomial in $n$ and $m$.
	For each $v$, let $p_v = |\{\mathcal{I}_{w^j} \mid 1 \le j \le m \land v \in  \mathcal{I}_{w^j}\}|$ and let $q_v = m - p_v$.
  	For an arbitrary $P \in \mathcal{P}$, we can rewrite $k_S(P)$ as follows:
  	\begin{align*}
  		k_S(P) &= \frac{1}{m} \sum_{j=1}^{m} k_M(\mathcal{I}_{w^j},P) 
  		= \frac{1}{m} \sum_{j=1}^{m} |\mathcal{I}_{w^j} \sym P|\\
  		&= \frac{1}{m} \sum_{j=1}^m |P\setminus \mathcal{I}_{w^j}| + |\mathcal{I}_{w^j} \setminus P|\\
  		&= \frac{1}{m} \left(\sum_{v \in P} q_v + \sum_{v\not\in P} p_v \right).
  	\end{align*}
  	A set $P \in \mathcal{P}$ minimizes the term $k_S(P) = \frac{1}{m} (\sum_{v \in P} q_v + \sum_{v\not\in P} p_v)$, if and only if, $q_v \le p_v$ holds for each $v \in P$.
  	Thus, we can compute the $P \in \mathcal{P}$ that minimizes $k_S(P)$ as follows:
  	\begin{enumerate}
  		\item Compute $q_v$ and $p_v$ for each $I_v \in \mathcal{I}$.
  		\item Return $P = \{v \in V \mid q_v\le p_v\}$.
  	\end{enumerate}
  	Since this algorithm can be executed in time polynomial in $n$ and $m$, the theorem follows.
\end{proof}

To conclude the section, we show the following lemma that compares the quality of the hypothesis sets $\hs$ and $\mathcal{P}$ w.r.t.~$k_M$.
It shows that, with respect to the expected $k_M$, the best prediction of $\mathcal{P}$ dominates the best prediction of $\hs$.

\begin{lemma}
	For any instance of one of our problems and any distribution $D$, it holds $\EX_{w \sim \ud}[k_M(w,\pred{w}^*) \ge \EX_{w \sim \ud}[k_M(\mathcal{I}_w,P^*)]$, where
	$\pred{w}^* = \arg\min_{\pred{w}' \in \hs} \EX_{w \sim \ud}[k_M(w,\pred{w}')]$ and $P^* =  \arg\min_{P' \in \mathcal{P}} \EX_{w \sim \ud}[k_M(\mathcal{I}_w,P')]$.
	Furthermore, there exists an instance and a distribution $D$ such that $\EX_{w \sim \ud}[k_M(w,\pred{w}^*)] > \EX_{w \sim \ud}[k_M(\mathcal{I}_w,P^*)]$.
\end{lemma}

\begin{proof}
	Let $\pred{w}^*$ and $P^*$ be as described in the lemma, and let $I_{\pred{w}^*}$ denote the set of intervals that are mandatory in the realization with precise weights $\pred{w}^*$.
	By definition, $\EX_{w \sim \ud}[k_M(w,\pred{w}^*)] = \EX_{w \sim \ud}[k_M(\mathcal{I}_w,\mathcal{I}_{\pred{w}^*})]$.
	Since $\mathcal{I}_{\pred{w}^*} \in \mathcal{P}$ and $P^* =  \arg\min_{P' \in \mathcal{P}} \EX_{w \sim \ud}[k_M(\mathcal{I}_w,P')]$, it follows $\EX_{w \sim \ud}[k_M(w,\pred{w}^*)] = \EX_{w \sim \ud}[k_M(\mathcal{I}_w,\mathcal{I}_{\pred{w}^*})]  \ge \EX_{w \sim \ud}[k_M(\mathcal{I}_w,P^*)]$.

	To show the second part of the lemma, consider the example of Figure~\ref{ex_domination}.
	The intervals of the figure define an interval graph with the vertices $\{v_1,\ldots,v_4\}$.	
	We assume that the precise weights for the different vertices	 are drawn independently.
	In this example, we have $P^* = \{v_1,v_3,v_5\}$.
	Since there is no combination of precise weights such that all members of $P^*$ are indeed mandatory, the in expectation best prediction of the precise weights $\pred{w}^*$ has either $\mathcal{I}_{\pred{w}^*} = \{v_1,v_3\}$ or $\mathcal{I}_{\pred{w}^*} = \{v_3,v_5\}$.
	Thus, $\EX_{w \sim \ud}[k_M(w,\pred{w}^*)] = 1.2601  > 1.2401 = \EX_{w \sim \ud}[k_M(\mathcal{I}_w,P^*)]$.
	\end{proof}
	
	\begin{figure}[tb]
			\centering
			\begin{tikzpicture}[line width = 0.3mm, scale = 0.75, transform shape]
			\interval{$I_{v_1}$}{0}{3}{0}{1}{1}		
			\interval{$I_{v_2}$}{2}{5}{1}{3.25}{5.5}		
			\interval{$I_{v_3}$}{4}{7}{0}{3.25}{5.5}					
			\interval{$I_{v_4}$}{6}{9}{1}{3.25}{5.5}	
			\interval{$I_{v_5}$}{8}{11}{0}{3.25}{5.5}	
			
			 \drawreal{1}{0}
			 \node[black!40!green] (1) at (1,0.5) {$1$};
			 
			 \drawreal{5.5}{0}
			 \node[black!40!green] (1) at (5.5,0.5) {$1$};
			 
			 \drawreal{10}{0}
			 \node[black!40!green] (1) at (10,0.5) {$1$};

			\drawreal{2.5}{1}
			\node[black!40!green] (1) at (2.5,1.5) {$0.51$};
			
			\drawreal{4.5}{1}
			\node[black!40!green] (1) at (4.5,1.5) {$0.49$};

			\drawreal{6.5}{1}
			\node[black!40!green] (1) at (6.5,1.5) {$0.49$};

			\drawreal{8.5}{1}
			\node[black!40!green] (1) at (8.5,1.5) {$0.51$};			
		\end{tikzpicture}
		\caption{Instance of the \nnew{sorting} problem for vertices $\{v_1,\ldots,v_5\}$ with a given distribution $D$ for the precise weights. 
		The circles illustrate the possible precise weights for the different intervals according to $D$. The associated numbers denote the probability that the interval has the corresponding precise weight.}
	\label{ex_domination}
	\end{figure}
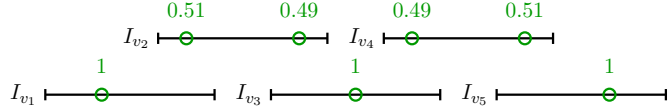

\section{Final remarks}
	We showed how 
	untrusted predictions enable us to circumvent known lower bounds for hypergraph orientation and sorting under explorable uncertainty
	and sparked the discussion on error measures by presenting a new error metric and showing relations between the different errors.
	As a next research step, we suggest investigating the application of error measure $k_M$ for different problems under explorable uncertainty.

\bibliographystyle{plain}
\bibliography{queries,ml}
\newpage 
\appendix

\section{NP-hardness of the offline problem for hypergraph orientation}
\label{app:overview}
\label{app:nphard}

\begin{theorem}
	\label{thm_min_nphard}
	It is NP-hard to solve the offline variant of hypergraph orientation under uncertainty.
\end{theorem}
\begin{proof}
	The proof uses a reduction from the vertex cover problem for 2-subdivision graphs, which is NP-hard~\cite{poljak74subdiv}.
	A 2-subdivision is a graph~$H$ which can be obtained from an arbitrary graph~$G$ by replacing each edge by a path of length four (with three edges and two new vertices).
	The graph in Figure~\ref{fig:2_subdiv} is a 2-subdivision of the graph in Figure~\ref{fig:base_graph}.
	
	\begin{figure*}[!ht]
		\centering
		\captionsetup[subfigure]{justification=centering}
		\begin{subfigure}{0.3\textwidth}
			\tikzstyle{every node}=[circle, draw, fill=black, inner sep=0pt, minimum width=4pt]
			\begin{tikzpicture}[thick, scale=0.75]
			\draw (0, 1) node[label=west:$a$]{} -- (1, 3) node[label=north west:$b$]{};
			\draw (1, 3) -- (5, 3) node[label=east:$c$]{};
			\draw (5, 3) -- (5, 1) node[label=east:$d$]{};
			\draw (1, 3) -- (5, 1);
			\draw (0, 1) -- (5, 1);
			\end{tikzpicture}
			\subcaption{}\label{fig:base_graph}
		\end{subfigure}\quad
		\begin{subfigure}{0.3\textwidth}
			\tikzstyle{every node}=[circle, draw, inner sep=0pt, minimum width=4pt]
			\begin{tikzpicture}[thick, scale=0.75]
			\draw (0, 1) node[fill=black, label=west:$a$]{} -- (1, 3) node[fill=black, label=north west:$b$]{};
			\draw (0.33, 1.67) node[fill=white,label=left:{$ab_1$}]{};
			\draw (0.67, 2.33) node[fill=white,label=left:{$ab_2$}]{};
			\draw (1, 3) -- (5, 3) node[fill=black, label=east:$c$]{};
			\draw (2.33, 3) node[fill=white,label=above:{$bc_1$}]{};
			\draw (3.67, 3) node[fill=white,label=above:{$bc_2$}]{};
			\draw (5, 3) -- (5, 1) node[fill=black, label=east:$d$]{};
			\draw (5, 1.67) node[fill=white,label=right:{$cd_2$}]{};
			\draw (5, 2.33) node[fill=white,label=right:{$cd_1$}]{};
			\draw (1, 3) -- (5, 1);
			\draw (2.33, 2.33) node[fill=white,label=below left:{$bd_1$}]{};
			\draw (3.67, 1.67) node[fill=white,label=above right:{$bd_2$}]{};
			\draw (0, 1) -- (5, 1);
			\draw (1.67, 1) node[fill=white,label=below:{$ad_1$}]{};
			\draw (3.33, 1) node[fill=white,label=below:{$ad_2$}]{};
			\end{tikzpicture}
			\subcaption{}\label{fig:2_subdiv}
		\end{subfigure}\\
		\begin{subfigure}{0.9\textwidth}
			\begin{tikzpicture}[line width = 0.3mm, scale=0.5]
			\intervalr{$I_a$}{0}{3}{0}{1.5}
			\intervalr{$I_b$}{6}{9}{0}{7.5}
			\intervalr{$I_c$}{12}{15}{0}{13.5}
			\intervalr{$I_d$}{18}{21}{0}{19.5}
			
			\intervalr{$I_{ab_1}$}{2}{5}{-1}{3.5}
			\intervalr{$I_{ab_2}$}{4}{7}{-2}{5.5}
			
			\intervalr{$I_{bc_1}$}{8}{11}{-1}{9.5}
			\intervalr{$I_{bc_2}$}{10}{13}{-2}{11.5}
			
			\intervalr{$I_{cd_1}$}{14}{17}{-1}{15.5}
			\intervalr{$I_{cd_2}$}{16}{19}{-2}{17.5}
			
			\intervalr{$I_{ad_1}$}{2}{12.2}{-3}{7.1}
			\intervalr{$I_{ad_2}$}{8.8}{19}{-4}{13.9}
			
			\intervalr{$I_{bd_1}$}{8}{14.6}{-5}{11.3}
			\intervalr{$I_{bd_2}$}{12.4}{19}{-6}{15.7}
			\end{tikzpicture}
			\subcaption{}\label{fig:minimum_reduction}
		\end{subfigure}
		\caption{NP-hardness reduction for the hypergraph orientation problem, from the vertex cover problem on 2-subdivision graphs. 
			(\subref{fig:base_graph}) A graph and (\subref{fig:2_subdiv}) its 2-subdivision.
			(\subref{fig:minimum_reduction}) The corresponding instance for the hypergraph orientation problem.}
		\label{fig_min_np_hard}
	\end{figure*}
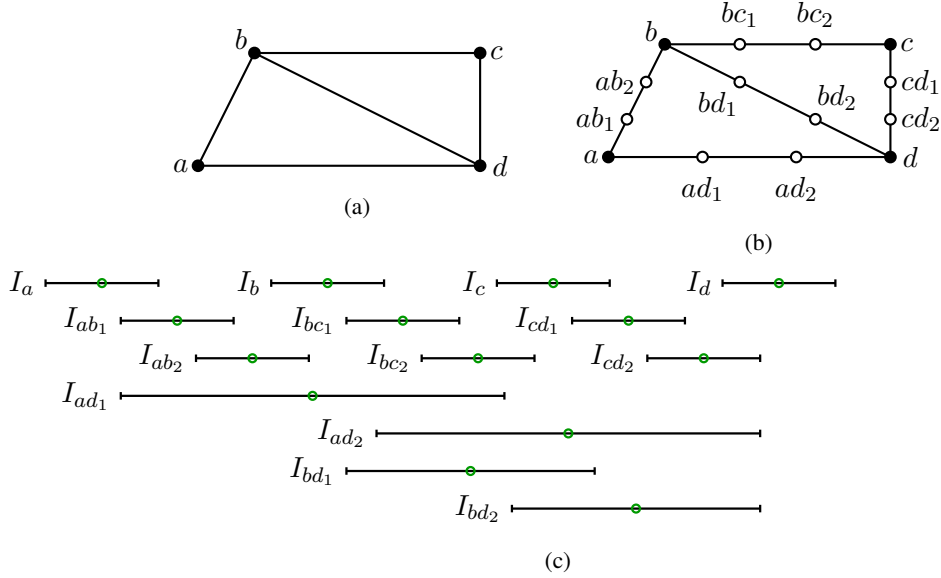
	
	Given a graph~$H$ which is a 2-subdivision of a graph~$G$. 
	We construct an instance of graph orientation under uncertainty $G'= (V',E')$ as follows.
	First, we set $G' = H$, i.e., we use the 2-subdivision as the input graph for our graph orientation instance.
	It remains to define the intervals and precise weights.
	
	For the original vertices of $G$, we define the intervals in a way such that no two intervals intersect each other.
	For an edge $e=\{u,v\}$ of $G$, let $v_1$ and $v_2$ denote the vertices introduced to replace $e$ when creating the 2-subdivision $H$.
	We define the intervals $I_{v_1}$ and $I_{v_2}$ such that two intervals of $I_v,I_{v_1},I_{v_2},I_u$ intersect if and only if the corresponding vertices define an edge in $G'=H$.
	As illustrated in Figure~\ref{fig:minimum_reduction}, placing the intervals in such a way is clearly possible.
	
	Finally, we can assign precise weights such that, for any edge $\{u,v\} \in E'$, neither $w_v \in I_u$ nor $w_u \in I_v$.
	By definition of the intervals, this is possible.
	Therefore, to orient each edge, we have to query at least one endpoint. 
	Furthermore, querying one endpoint is sufficient to orient any edge.
	Clearly every solution to the offline variant of the graph orientation problem corresponds to a vertex cover of~$H$, and vice-versa.
\end{proof}

\end{document}